%% file: arxiv-main.tex
\documentclass[11pt]{article}%

\usepackage{xspace}
\usepackage{amsfonts, amsthm, amsmath}
\usepackage[in]{fullpage}%

\newtheorem{theorem}{Theorem}[section]

\newtheorem{lemma}[theorem]{Lemma}
\newtheorem{proposition}[theorem]{Proposition}
\newtheorem{corollary}[theorem]{Corollary}
\newtheorem{remark}[theorem]{Remark}

\newtheorem{definition}[theorem]{Definition}
\usepackage{scalerel}
\usepackage{needspace}
\usepackage{caption}%
\input{arxiv-prefix}%

\title{Average-Case Communication Complexity of Statistical Problems}
\author{%
  Cyrus Rashtchian%
  \thanks{Department of Computer Science \& Engineering, UC San Diego. %
      \url{crashtchian@eng.ucsd.edu} }%
  \and %
  David P. Woodruff%
  \thanks{Computer Science Department, Carnegie Mellon University.  
      \url{dwoodruf@cs.cmu.edu} }%
    \and %
    Peng Ye%
    \thanks{Institute for Interdisciplinary Information Sciences, Tsinghua University. %
  	\url{yep17@mails.tsinghua.edu.cn}}%
  \and %
  Hanlin Zhu%
  \thanks{Institute for Interdisciplinary Information Sciences, Tsinghua University. %
  	\url{zhuhl17@tsinghua.org.cn}}%
}

\begin{document}

\maketitle

\begin{abstract}%
We study statistical problems, such as planted clique, its variants, and sparse principal component analysis in the context of average-case communication complexity. Our motivation is to understand the statistical-computational trade-offs in streaming, sketching, and query-based models. Communication complexity is the main tool for proving lower bounds in these models, yet many prior results do not hold in an average-case setting. We provide a general reduction method that preserves the input distribution for problems involving a random graph or matrix with planted structure. Then, we derive two-party and multi-party communication lower bounds for detecting or finding planted cliques, bipartite cliques, and related problems. As a consequence, we obtain new bounds on the query complexity in the edge-probe, vector-matrix-vector, matrix-vector, linear sketching, and $\mathbb{F}_2$-sketching models. Many of these results are nearly tight, and we use our techniques to provide simple proofs of some known lower bounds for the edge-probe model. 
\end{abstract}

\input{intro.tex}

\input{PC-COLT.tex}

\section{Bipartite Planted Clique Detection}
\input{BPC}

\subsection{Lower Bound for Finding a Planted Biclique}
\input{FindBPC}

\section{Semi-Random Planted Clique}
\input{SRPC}

\section{Promise Planted Clique}
\input{PPC}

\section{Hidden Hubs}
\input{HH}

\section{Sparse Principal Component Analysis}
\input{SPCA.tex}

\section{Conclusion}
\input{conclusion.tex}

\paragraph{Acknowledgments.}
D. Woodruff would like to thank NSF grant No. CCF-181584, Office of Naval Research (ONR) grant N00014-18-1-256, and a Simons Investigator Award.

\linespread{.98}
\small
\bibliographystyle{alpha}%
\bibliography{references}



\end{document}

%% file: arxiv-prefix.tex
\usepackage{mleftright}%
\usepackage{amsmath}%
\usepackage{amssymb}
\usepackage{mathtools}
\usepackage{booktabs}
 \usepackage{enumerate}
\usepackage{multirow}
\usepackage{multicol}

\usepackage{mathrsfs}
\usepackage{xcolor}

\usepackage{hyperref}%
\definecolor{bluish-green}{HTML}{17747A}
\definecolor{red-purple}{HTML}{72177A}
\hypersetup{%
   breaklinks,%
   colorlinks=true,%
   urlcolor=[rgb]{0.25,0.0,0.0},%
   linkcolor=red-purple,%
   citecolor=bluish-green,%
   filecolor=[rgb]{0,0,0.4}, anchorcolor=[rgb]={0.0,0.1,0.2}%
}

\usepackage{titlesec}%
\titlelabel{\thetitle. }%

\renewcommand{\epsilon}{\varepsilon}
\usepackage{booktabs}
\usepackage{color, colortbl}
\renewcommand{\vec}{\mathsf{vec}}
\newcommand{\emvar}{\widehat{\mathsf{var}}}

\numberwithin{equation}{section}

\newcommand{\HLinkShort}[2]{\hyperref[#2]{#1\ref*{#2}}}
\newcommand{\HLink}[2]{\hyperref[#2]{#1~\ref*{#2}}}
\newcommand{\HLinkPage}[2]{\hyperref[#2]{#1~\ref*{#2}%
      $_\text{p\pageref{#2}}$}}
\newcommand{\HLinkPageOnly}[1]{\hyperref[#1]{Page~\refpage*{#1}%
      $_\text{p\pageref{#1}}$}}

\newcommand{\HLinkSuffix}[3]{\hyperref[#2]{#1\ref*{#2}{#3}}}
\newcommand{\HLinkPageSuffix}[3]{\hyperref[#2]{#1\ref*{#2}%
      #3$_\text{p\pageref{#2}}$}}

\newcommand{\seclab}[1]{\label{sec:#1}}
\newcommand{\secref}[1]{\HLink{Section}{sec:#1}}

\newcommand{\corlab}[1]{\label{cor:#1}}
\newcommand{\corref}[1]{\HLink{Corollary}{cor:#1}}%

\newcommand{\remlab}[1]{\label{rem:#1}}
\newcommand{\remref}[1]{\HLink{Remark}{rem:#1}}%

\newcommand{\lemlab}[1]{\label{lemma:#1}}
\newcommand{\lemref}[1]{\HLink{Lemma}{lemma:#1}}%

\newcommand{\tablab}[1]{\label{table:#1}}%
\newcommand{\tabref}[1]{\HLink{Table}{table:#1}}%

\newcommand{\proplab}[1]{\label{prop:#1}}

\newcommand{\propref}[1]{\HLink{Proposition}{prop:#1}}

\newcommand{\thmlab}[1]{{\label{theo:#1}}}
\newcommand{\thmref}[1]{\HLink{Theorem}{theo:#1}}

\providecommand{\eqlab}[1]{}%
\renewcommand{\eqlab}[1]{\label{equation:#1}}

\newcommand{\Eqref}[1]{\HLinkSuffix{Eq.~(}{equation:#1}{)}}

\IfFileExists{.latex_printer_friendly}{\def\GenPrinterVer{1}}{}%

\ifx\GenPrinterVer\undefined
   \IfFileExists{.latex_color}{\def\GenColorMath{1}}{}
\else
\fi

\ifx\GenColorMath\undefined
\else
\fi

\newcommand{\eps}{{\varepsilon}}%

\newcommand{\wt}{\widetilde}%

\usepackage{stmaryrd}%

\definecolor{blue25}{rgb}{0, 0, 11}

\newcommand{\remove}[1]{}

\DefineNamedColor{named}{AlgorithmColor}{rgb}{0.1,0,0.15}

\newcommand{\M}{\boldsymbol{M}}
\newcommand{\A}{\boldsymbol{A}}
\newcommand{\B}{\boldsymbol{B}}

\newcommand{\X}{\boldsymbol{X}}
\newcommand{\Y}{\boldsymbol{Y}}

\newcommand{\V}{\boldsymbol{V}}
\newcommand{\Z}{\boldsymbol{Z}}
\newcommand{\I}{\boldsymbol{I}}

\newcommand{\T}{\mathrm{T}}
\renewcommand{\u}{\boldsymbol{u}}
\renewcommand{\v}{\boldsymbol{v}}
\newcommand{\w}{\boldsymbol{w}}
\newcommand{\x}{\boldsymbol{x}}
\newcommand{\y}{\boldsymbol{y}}

\newcommand{\F}{\mathbb{F}}

\newcommand{\R}{\mathbb{R}}
\newcommand{\uMv}{\mathsf{u^TMv}}
\newcommand{\Mv}{\mathsf{Mv}}

\newcommand{\PC}{\textsf{PC}\xspace}
\newcommand{\BPC}{\textsf{BPC}\xspace}
\newcommand{\SRPC}{\textsf{SRPC}\xspace}
\newcommand{\PPC}{\textsf{PPC}\xspace} 
\newcommand{\HH}{\textsf{HH}\xspace}   
\newcommand{\PE}{\textsf{PE}\xspace}
\newcommand{\SPCA}{\textsf{SPCA}\xspace}

\newcommand{\SCDC}{\textsf{SCDC}\xspace}

\makeatletter
\let\c@table\c@figure
\makeatother


%% file: intro.tex
\section{Introduction}
\seclab{intro}

The planted clique and sparse principal component analysis problems embody an enduring interest in computational vs.~statistical trade-offs. These problems have the intriguing property that
it may be easy to detect the presence of planted structure in super-polynomial time. However, it is a central open problem to determine whether an efficient solution exists, even though we know that one is possible information theoretically~\cite{abbe2017community, bandeira2018notes, berthet2013complexity, jordan2015machine}.

The planted clique problem involves distinguishing between two distributions on $n$-vertex graphs. In the first, the graph is generated from the Erdos-Renyi model $G(n,1/2)$, where each edge is independently present with $1/2$ probability. The second distribution $G(n,1/2,k)$ has a $k$-clique planted in a random subset of $k$ vertices, and the remaining edges exist independently with $1/2$ probability. From an information theoretic point of view, detection is possible if $k \geq (2+\delta)\log_2 n$ for any constant $\delta >0$ because the largest clique in a random graph has size $(2+o(1))\log_2 n$ almost surely. When the clique is very large, i.e., $k \gg \sqrt{n}$, many methods can distinguish the two distributions in polynomial time (and find the planted clique). However, when $k = o(\sqrt{n})$, all known algorithms require super-polynomial time
~\cite{alon1998finding,arias-castro2014,dekel2011finding,feige2000finding,frieze2008new, kuvcera1995expected,  ma2015computational}.

A natural question is to understand the complexity of statistical problems in other models. Query-based algorithms form the basis of sublinear time methods for massive graphs~\cite{avrachenko2014, leskovec2006, maiya2010,sound2017}. In network monitoring applications, streaming and sketching algorithms are used for real-time data analytics when the graph is too large to fit into memory or when the edges arrive over time~\cite{ahmad2017,gupta2016}. Typical network activity could be modeled as a distribution over edge connections, and the presence of some planted subgraph structure could signify anomalous or suspicious group behavior~\cite{chandola2009, huang2015}. This motivates understanding the query and streaming complexity of detection problems under average-case distributions.

R\'{a}cz and Schiffer consider the edge-probe model, which measures the number of edge existence queries to solve a problem~\cite{racz2019finding}; this model is also known as the dense graph model~\cite{goldreich2017introduction,goldreich2009algorithmic,goldreich1998property}. Here, there are no computational constraints, making it feasible to study clique detection when $k = o(\sqrt{n})$. R\'{a}cz and Schiffer show that $\widetilde \Theta(n^2/k^2)$ edge-probe queries are necessary and sufficient to detect a planted $k$-clique, and they also prove similar bounds for finding the clique~\cite{racz2019finding}.  A more general model involves linear sketches~\cite{woodruff2014sketching}. Representing an $n \times n$ matrix $\A$ as a vector $\vec(\A)$ with $n^2$ entries, a query returns $\u^\top \vec(\A)$ for a vector $\u$ with polynomially-bounded entries. 
A restriction of the sketching model, the $\uMv$ model, returns $\u^\top \M \v$ for vectors $\u,\v$, where $\M$ is an unknown matrix~\cite{rashtchian2020vector}. This specializes the $\Mv$ model, which returns $\M \v$~\cite{sun2019querying}. These models all generalize edge-probes. 

Since there are advantages and disadvantages to these various types of queries, it is often worthwhile to understand the complexity of solving certain problems in each of the models. If an algorithm can be implemented in a more restricted model, then it may be more useful in practice. On the other hand, a lower bound for a more general model would imply the same bound for any specialized model. To this end, it is common to prove lower bounds on the communication complexity~\cite{kushilevitz2006communication, rao_yehudayoff_2020}. Then, by showing that a query-efficient or space-efficient algorithm can solve a communication problem, lower bounds can be derived for the query complexity.

\subsection{Our Results}
We provide a general method to encode a communication game as a statistical graph or matrix problem while retaining the input distribution. In the next subsection, we provide technical details about how to execute this approach. Here, we summarize our query complexity and communication upper and lower bounds. Throughout, we often assume that $k = o(\sqrt{n})$ because otherwise there is often an $O(1)$ query upper bound (see~\secref{upper-bounds}).
\begin{itemize}
\item In \secref{pc-edge}, we provide an alternate proof of the R\'{a}cz-Schiffer bound showing that detecting a planted $k$-clique requires $\Omega(n^2/k^2)$ edge-probe queries~\cite{racz2019finding}. Then, we investigate whether stronger models are able to succeed with fewer queries. Our most technical contribution shows that $\widetilde \Omega(n^2/k^4)$ queries are necessary to detect a planted $k$-clique in the linear sketching model (and hence also in the $\uMv$ model); this appears in \secref{pc}, and it follows from an information complexity argument in \secref{IC}. The linear sketching model is more powerful, in general, than the edge-probe model, and we leave open the question of whether the query complexity is $\widetilde \Omega(n^2/k^4)$ or $\widetilde O(n^2/k^2)$ or somewhere in between.

\item We also consider detecting and finding a planted $k \times k$ bipartite clique (biclique). When we have to output the planted biclique, we provide nearly tight upper and lower bounds in the $\Mv$ model. When $k = o(\sqrt{n})$, it is easy to see that $O(\frac{n}{k})$ $\Mv$ queries suffice (\secref{upper-bounds}), and we prove that $\Omega(\frac{n}{k \log n})$ queries are necessary to find a planted $k \times k$ biclique (\secref{findbpc}).  We also obtain trade-offs in other query models, depending on the clique size, where we generally consider a planted $r\times s$ biclique. 
We provide an $\widetilde \Omega(n^2/(r^2 s^2))$ lower bound for general linear sketching. To complement this, we exhibit an algorithm in the $\uMv$ model that uses only $\widetilde O(n^2/(r^2 s))$ queries, assuming that $r \gg \sqrt{n \log n}$. Our algorithm borrows ideas from CountSketch~\cite{charikar2002}, as high-degree planted vertices can be considered as $\ell_2$ heavy hitters. Finally, we give a stronger $\widetilde \Omega(n^2/(r s))$ lower bound in the edge-probe model, which is tight up the logarithmic factors (\secref{bpc}). 

\item We further uncover qualitatively different trade-offs by considering variants of the planted clique detection problem.  We investigate the sandwich semi-random version of planted clique from~\cite{feige2000finding}. In this model, an adversary is allowed to remove some number of edges that are not part of the planted clique. For this variant, we prove that $\widetilde \Theta(n^2/k^2)$ bits are required and sufficient for a related communication game (\secref{srpc}). The complexity in the linear sketching model is $\widetilde \Theta(n^2/k^2)$, where the upper bound follows from existing algorithms in the edge-probe model. This indicates that any improved algorithm for the usual planted clique problem would require non-trivial algorithmic techniques.

\item Then, we study a {\em promise} variant. If the players know that the planted clique occurs in one of $O(n^2/k^2)$ edge-disjoint subgraphs, then  $\widetilde \Theta(n^2/k^4)$ bits of communication are both necessary and sufficient for detection (\secref{ppc}). This shows that the $k^4$ dependency is tight in this promise variant. While our motivation is technical, related promise problems have been studied for other average-case reductions~\cite{brennan2020reducibility} and for network inference when prior information has been previously obtained~\cite{sound2017}. 

\item Finally, we also provide lower bounds for the hidden hubs problem~\cite{pmlr-v65-kannan17a} (\secref{hh}) and for sparse PCA~\cite{berthet2013complexity} (\secref{spca}).

\item Our edge-probe lower bounds extend to the $\F_2$ sketching model, where querying with a vector $\u$ returns the value $\u^{\T} \vec(\A)$ over $\F_2$, where again $\A$ is the adjacency matrix. While we do not know a separation between these models for finding planted structure, the $\F_2$ sketching model is a formal generalization of the edge-probe model (using a standard basis vector as the query).  Our results also immediately provide upper and lower bounds for streaming algorithms, but we focus on communication and query complexity for brevity.
\end{itemize}

\begin{table}[t] 	\renewcommand{\arraystretch}{1.5} \footnotesize
	\centering 
	\caption{Average-case query complexity for statistical problems (\secref{pre} has definitions). We suppress $\mathrm{polylog}(n)$ factors and assume $k = o(\sqrt{n})$ and $k = \Omega(\log n)$ and $r = \Omega( \sqrt{n \log n})$.}
	\begin{tabular}{l|c|c|c|l} 
		\toprule 
		& $\F_2$ sketch \& edge-probe & Linear sketch \& $\uMv$ & $\Mv$ & Ref. \\
		\midrule
		\PC & $\wt \Theta(n^2/k^2)$ & $\wt \Omega(n^2/k^4) \ \  \wt O(n^2/k^2)$ & $\wt \Omega(n/k^4) \ \  O(n/k)$ &  \secref{pc-edge} \& \ref{sec:pc} \\
		\BPC & $\wt \Theta(n^2/(rs))$ & $\wt \Omega(n^2/(r^2s^2))\ \ \wt O(n^2/(r^2 s))$ & $\wt \Omega(n/(r^2s^2)) \ \ \wt O(n/\min(r,s))$ & \secref{bpc} \\
		\textsc{Find}\BPC & $\wt \Theta(n^2/k^2)$ & $\wt \Omega(n^2/k^4) \ \ \wt O(n^2/k^2)$ & $\wt \Theta(n/k)$ & \secref{findbpc}\\
		\SRPC & $\wt \Theta(n^2/k^2)$ & $\wt \Theta(n^2/k^2)$ & $\wt \Omega(n/k^2) \ \ O(n/k)$ & \secref{srpc}\\
		\PPC & $\wt \Theta(n^2/k^2)$ & $\wt \Theta(n^2/k^4)$ & $\wt \Omega(n/k^4) \ \ \wt O(\min(\frac{n}{k},\  \frac{n^2}{k^4}))$  & \secref{ppc}\\
		\HH & $\wt \Omega(n^2/k^4)$\ \ $\wt O(n^2/k)$ & $\wt \Omega(n^2/k^4) \ \ \wt O(n^2/k)$  & $\wt \Omega(n/k^4) \ \ \wt O(n/k)$ & \secref{hh}\\
		\SCDC & $\Omega\left(\frac{k^2}{\theta^2}\right)$ & $\widetilde{\Omega}\left(\frac{k^4}{t^2\theta^4}\right)$ & $\widetilde{\Omega}\left(\frac{k^4}{t^3\theta^4}\right)$ & \secref{spca}\\	
		\bottomrule
	\end{tabular}
	\tablab{results}
\end{table}

\subsection{Technical Overview}
\subsubsection*{Lower Bounds in the general linear sketching and $\uMv$ models}
We start by describing our lower bound techniques for the planted clique problem. The average case notion of our problems makes reductions from standard problems in communication complexity, such as multi-player set disjointness, non-trivial, as they do not give us instances from our desired distribution. This is unlike existing worst-case clique communication lower bounds~\cite{braverman2018new, halldorsson2012streaming}, which reduce directly from set disjointness. 

We instead use a communication complexity model that allows the players to have access to shared public randomness,  as well as private randomness. Then, we consider a multi-player hypothesis testing problem, introduced in \cite{braverman2016communication}, where each player either receives an independent sample from a distribution $\mu_0$ or a distribution $\mu_1$ and the players would like to decide which case they are in. Using a strong data processing inequality, the information cost of such a protocol was shown to be $\Omega(1)$ if $\mu_0 \geq \frac{1}{c} \mu_1$ for a constant $c > 0$, even when information is measured with respect to $\mu_0$ alone~\cite{braverman2016communication}. We combine this with the information complexity framework of \cite{bar2004information} to prove a direct sum theorem for solving the OR of multiple copies of this problem (here, the ``OR'' of many instances evaluates to true whenever at least one of the component instances evaluates to true). We guarantee when the OR evaluates to 1, then exactly one copy is from $\mu_1$. We note that Weinstein and Woodruff~\cite{weinstein2015simultaneous} prove a distributional result for \emph{simultaneous} multi-party communication, but this would only apply to non-adaptive query algorithms, whereas our results apply even to adaptive query algorithms. Moreover, the distributions considered in~\cite{weinstein2015simultaneous} are specific, and not the same as the ones we need for our applications, which we now discuss. 

The main remaining task is to choose distributions $\mu_0$ and $\mu_1$ so that the resulting multi-copy distribution matches that of the planted clique problem. We first use a clique partitioning scheme of \cite{conlon2012short} which although related to results on proving worst-case clique communication lower bounds \cite{braverman2018new, halldorsson2012streaming}, does not seem to have been used before in this context. This gives us $\Omega(n^2/k^2)$ edge-disjoint cliques on $k$ vertices each. We have $\binom{k}{2}$ players, and each is assigned one edge from each clique. We let $\mu_0$ be the uniform distribution, so that if the OR of the $\Omega(n^2/k^2)$ instances above is 0, we have a graph from $G(n,1/2)$. Otherwise the OR evaluates to 1 (i.e., the OR being true corresponds to having at least one planted $k$-clique). We let $\mu_1$ be the constant distribution with value $1$ (i.e., the value is always 1 in these positions), and we randomly permute vertex labels, so that in this case we have exactly one planted clique on $k$ vertices and otherwise have a $G(n,1/2)$ instance.   By our choice of $\mu_0$ and $\mu_1$, this gives us an $\Omega(1)$ information cost lower bound per copy, an $\Omega(n^2/k^2)$ lower bound for the OR problem, and an $\tilde{\Omega}(n^2/k^4)$ $\uMv$ query lower bound by simulating each query across $\Theta(k^2)$ players. 

\subsubsection*{From the $\uMv$ to the $\Mv$ model}
While the above communication game can be applied to the $\Mv$ model, it would only give us an $\Omega(n/k^4)$ lower bound. We strengthen this to a nearly {\em optimal} $\widetilde{\Omega}(n/k)$ lower bound for the related planted bipartite clique (biclique) problem, and where the algorithm is promised to return a $k \times k$ biclique when it exists (later, we more generally consider $r \times s$ bicliques, but for this discussion, we let $k = r = s$). The issue is the algorithm retrieves too much information ($\Omega(n)$ bits) with each $\Mv$ query.
To get around this, we only consider inputs when there actually exists a randomly planted biclique. Although the distinguishing problem is trivial now (we always have a biclique), since the algorithm must return the vertices in the biclique, it still has a non-trivial task. 

Next, we fix the set of $k$ left vertices in the biclique. They are random and form a valid input instance, but they are known to the algorithm. This might seem counterintuitive, as it only helps the algorithm. Also, each $\mathsf{Mv}$ query only reveals $O(k \log n)$ bits of information, and thus we will only pay an $O(k \log n)$ factor instead of an $O(k^2 \log n)$ factor per query, in our query to communication simulation. The next idea is to also fix $k-1$ of the vertices on the right in the biclique; they are again random, and so they form a valid input instance, but they are known to the algorithm. Again, this might seem counterintuitive, but it helps our analysis because now it gives us $n-(k-1) = \Omega(n)$ possible remaining vertices in the right part, and any of them can be the potential last right vertex. This gives us $\Omega(n)$ possible cliques rather than $\Omega(n/k)$ if we were to partition the right vertices into vertex-disjoint bicliques, which we would need in order to have edge-disjoint bicliques, since we have already fixed the vertices in the left of the biclique. We show the algorithm needs to reveal $\Omega(n)$ bits of information to figure out the missing right vertex. Since each player now corresponds to a row in this $k \times n$ matrix, when we do the query to communication simulation we obtain an $\widetilde{\Omega}(n/k)$ overall lower bound, losing one factor of $k$ for the number of players, and a factor of $O(\log n)$ to transmit its dot product with the query vector. The full details are in \secref{findbpc}.


\subsection{Related Work}

Even though average-case reductions have been studied for many models, there seems to be a gap in our understanding for communication complexity. In the context of graph problems, current techniques usually construct a hard instance by identifying specific families of graphs that encode a communication game. Unfortunately, this does not answer the question of whether the complexity remains high when we are in a hypothesis testing setting; e.g., {\sc max-clique} is NP-Hard, but distinguishing between $G(n,1/2)$ or $G(n,1/2,k)$ for $k \geq (2+\delta)\log_2 n$ can be solved in quasi-polynomial time by checking all cliques of size $O(\log n)$. In the realm of communication complexity, known results for approximating the size of the maximum clique do not match the distributions for the planted clique problem~\cite{braverman2018new, halldorsson2012streaming}. At a high level, we also use a combination of (nearly) covering the graph with subsets of vertices and then reducing to a version of set disjointness. However, we invoke a slightly different graph decomposition, ensuring that the subsets intersect in at most one vertex, and moreover, the players can use public randomness to match the input distribution of the average-case statistical problems. 

The conjectured hardness of detecting/finding cliques has been used to derive statistical vs.~computational trade-offs for many average-case problems~\cite{brennan2019average,berthet2013complexity, boix2019average, kunisky2019notes}. Thus, an open direction from our work is whether analogous reductions can extend our communication complexity lower bounds to other statistical problems (e.g., the stochastic block model) or property testing in the dense graph model~\cite{goldreich2017introduction}.  Our study of the promise planted clique problem is inspired by conjectures regarding the secret linkage of prior information of the planted clique~\cite{brennan2020reducibility}. 
Recent work studies the query complexity of approximating the maximum clique and/or finding the clique in the $G(n,1/2)$ model~\cite{alweiss2019subgraph,feige2020finding, mardia2020finding}. Lower bounds for the planted clique problem have been shown for the statistical query model~\cite{feldman2017statistical} and for sum-of-squares~\cite{barak2019nearly,meka2015sum}.


\section{Preliminaries}
\seclab{pre}

Let $[n]= \{1, 2, \ldots, n\}$. 
We use {\em with high probability} to mean $1-O(1/n^c)$ for a constant $c > 0$ and {\em with constant probability} to mean at least $9/10$. The notation $\widetilde O, \widetilde \Omega, \widetilde \Theta$ hides $\mathrm{polylog}(n)$ factors.

\subsection{Problems, Games, and Query Models}

For each problem, we define the null hypothesis $H_0$ and the alternate hypothesis $H_1$. The goal is to determine which hypothesis a graph or matrix has been drawn from with constant probability. 
\begin{itemize}
\item {\bf Planted Clique (\PC).}  
The input is a graph with $n$ vertices, described as an $n\times n$ adjacency matrix $\A$. For $H_0$ each edge occurs with probability $1/2$, i.e., $\A \sim G(n, 1/2)$ in the Erd\'{o}s-Renyi model. For $H_1$ there is a planted $k$-clique, i.e., a set $R$ is randomly chosen over all size~$k$ subsets of $[n]$; first $\A \sim G(n, 1/2)$, then we set $A_{ij}$ to 1 for all $i,j\in R$ with $i \neq j$.

\item {\bf Bipartite Planted Clique (\BPC).} 
The input is a bipartite graph with $n$ vertices on each side, described as an $n\times n$ matrix $\A$. For $H_0$ each edge occurs with probability $1/2$, i.e., each entry of $\A$ is sampled from $\texttt{Bernoulli}(1/2)$. For $H_1$ there is an $r\times s$ planted biclique, i.e., two sets $R$ and $S$ are randomly chosen over all size $r$ subsets of $[n]$ and all size $s$ subsets of $[n]$ respectively, and then $A_{ij}$ is set to 1 for all $i\in R$ and $j\in S$, and all the remaining entries follow $\texttt{Bernoulli}(1/2)$ independently.

\item {\bf Semi-Random Planted Clique (\SRPC).} The semi-random model was introduced by Blum and Spencer \cite{BLUM1995204}. There are variants of the semi-random model, and we specifically consider the sandwich model \cite{743518}. In this model, there is an adversary which can remove arbitrary edges outside the planted clique. We describe our hypothesis testing problem as follows:
\begin{itemize}
    \item[$H_0$:] the adversary chooses any graph $G^*$ such that $G_{\min}\subseteq G^*\subseteq G_{\max}$, where $G_{\max}$ is a random graph drawn from $G(n, 1/2)$ and $G_{\min}$ is the empty graph.
    \item[$H_1$:] the adversary chooses any graph $G^*$ such that $G_{\min}\subseteq G^*\subseteq G_{\max}$, where $G_{\max}$ is a random graph drawn from $G(n, 1/2, k)$ and $G_{\min}$ only contains the planted clique.
\end{itemize}

\item {\bf Promise Planted Clique (\PPC).} 
There is a fixed and known collection $S$ of subsets of $k$ vertices such that every pair of subsets intersects in at most one vertex. For $H_0$, the graph is $G(n,1/2)$ as in the $\PC$ problem. For $H_1$, the planted clique is chosen from $S$. Clearly, $|S|\le n^2/k^2$, and if $|S| = \Theta(n^2/k^2)$, then $k \leq O(\sqrt{n})$. The motivation for the \PPC problem is that we use a graph decomposition result to define $S$ for some of our reductions (see \secref{clique_partition}). Thus, the \PPC problem captures the relative difficulty of the problem when the set of possible cliques is known in advance. From an algorithmic point of view, this makes the problem trivial. On the other hand, from a query complexity point of view, our upper and lower bounds for sketching algorithms nearly match for the \PPC problem.

\item {\bf Hidden Hubs (\HH).} In the hidden hubs model $H(n, k, \sigma_0, \sigma_1)$, an $n\times n$ random matrix $\A$ is generated as follows~\cite{pmlr-v65-kannan17a}. First randomly choose a subset $S$ of $k$ rows. Entries in rows outside $S$ are generated from the Gaussian distribution $p_0 = \mathcal{N}(0, \sigma_0^2)$. For each row in $S$, choose $k$ entries to be generated from $p_1 = \mathcal{N}(0, \sigma_1^2)$, and the other $n-k$ entries from $p_0$. The hypothesis testing problem ($\HH$ problem) is to distinguish $H_0$ and $H_1$, where $H_0$ is an $n\times n$ random matrix with all entries generated from $\mathcal{N}(0, \sigma_0^2)$, and $H_1$ is the model $H(n, k, \sigma_0, \sigma_1)$. 

\item {\bf Sparse Component Detection Challenge (\SCDC).}
We consider the sub-Gaussian version of the Sparse Principal Component Analysis (\SPCA) problem, using elements of a known reduction from the \PC problem~\cite{berthet2013complexity}.
The empirical variance of $t$ vectors $\X_1,\ldots,\X_t \in \mathbb{R}^d$ in direction $\v$ is defined as 
$\emvar(\v) = \frac{1}{t} \sum_{i=1}^t (\v^\top \X_i)^2.
$ 
Let $\theta$ and $k$ be parameters, and let $\zeta \in (0,1)$ be a fixed, small constant. 
We let $\mathcal{D}_0$ denote the set of product distributions over $t$ i.i.d.~vectors $\X_1,\ldots,\X_t \in \mathbb{R}^d$ such that for all unit vectors $\v$ we have
\begin{equation}\eqlab{spca-d0-intro}
\Pr\left[|\emvar(\v) - 1 | > 4\sqrt{\frac{\log(2/\zeta)}{d}} + 4\frac{\log(2/\zeta)}{d}\right] \leq \zeta.
\end{equation}
In other words, $\mathcal{D}_0 := \{\mathbf{P}_0 \mid \Eqref{spca-d0-intro} \mbox{ holds}\}$, and $\mathcal{D}_0$ contains, e.g., isotropic distributions.

We let $\mathcal{D}_1^{k,\theta}$ denote the set of product distributions over $t$ i.i.d.~vectors $\X_1,\ldots,\X_t \in \mathbb{R}^d$ such that for all unit vectors $\v$  with at most $k$ nonzero entries ($\|\v\|_0 \leq k$), we have
\begin{equation}\eqlab{spca-d1-intro}
\Pr\left[\left(\emvar(\v) - (1+\theta)\right) < -2\sqrt{\frac{\theta k \log(2/\zeta)}{d}} - 4\frac{\log(2/\zeta)}{d}\right] \leq \zeta.
\end{equation}
Similarly, $\mathcal{D}_1^{k,\theta} := \{\mathbf{P}_1 \mid \Eqref{spca-d1-intro} \mbox{ holds}\}$.
Then, for the \SCDC problem, we define two hypotheses to test; the inputs $\X_1,\ldots,\X_t$ are drawn from $\mathbf{P}$ such that
$$
\mathcal{H}_0: \X_1,\ldots,\X_t \sim \mathbf{P}_0 \in \mathcal{D}_0 
\qquad\qquad \mbox{vs.} \qquad\qquad
\mathcal{H}_1: \X_1,\ldots,\X_t \sim \mathbf{P}_1 \in \mathcal{D}_1^{k,\theta}.
$$
Our goal is to distinguish which family of distributions $\X_1, \ldots, \X_t$ is sampled from. The motivation is that $\mathcal{D}_0$ contains $ \mathcal{N}(0, \I_d)$ and $\mathcal{D}_1$ contains $\mathcal{N}(0, \I_d+\theta\u\u^{\T})$ when $\u$ is a $k$-sparse unit vector. Hence, this generalizes the spiked covariance model~\cite{berthet2013optimal, brennan2019average}.


\end{itemize} 

\paragraph{Communication Games.} Throughout we use {\em problem} to refer to the detection problems above, and we use {\em game} to refer to the analogous communication complexity problem. For \PC, \SRPC, and \BPC we define the games as follows for $t \geq 2$ players. The players receive edge-disjoint subgraphs of a graph $G$ such that the union of the edges equals the whole graph. Equivalently, the players receive $n \times n$ adjacency matrices corresponding to their subset of the edges, and they must solve the corresponding problem on the graph defined by the {\em sum} of the adjacency matrices (which is the adjacency matrix of the whole graph since the edge sets are disjoint). The players are promised that $G$ is either drawn from $H_0$ or $H_1$ as in the problems defined above. To succeed, the players must determine which distribution $G$ is drawn from with constant probability.  For $\PPC$, the only difference is that the players also all know the set $S$ of possible locations for the planted clique. For $\HH$, the players instead receive $n \times n$ matrices with disjoint supports, where the sum of these matrices is drawn from one of the two hypotheses. 
While many of these games have been defined for the union of the graphs, we also make use of an XOR variant for the 2-player version of the games. More precisely, Alice and Bob each receive adjacency matrices $G_1$ and $G_2$, which are not necessarily disjoint in the support of their entries. Then, they must solve the corresponding problem on the graph $G_1 \oplus G_2$, where an edge is present in $G_1 \oplus G_2$ if and only if it is present in exactly one of $G_1$ or $G_2$. In other words, we use the XOR of the adjacency matrices. This variant will be used for proving $\F_2$ sketching lower bounds (which will immediately imply the edge-probe lower bounds).

\paragraph{Query Models.} 
Matrices and vectors have polynomially bounded integer entries. Let $\vec(\A)$ denote the vectorization of an $n \times n$ matrix $\A$, i.e., $n^2$ entries listed in a fixed order.

\begin{itemize}
\item {\bf Edge-Probe Model.} Querying position $(i,j)$ returns $A_{ij}$.
\item {\bf $\uMv$ Model.} Querying with vectors $\u,\v \in \R^n$ returns  $\u^{\T} \A \v$ over $\R$.
\item {\bf $\Mv$ Model.} Querying with a vector $\v \in \R^n$ returns  $\A \v$ over $\R$.
\item {\bf $\F_2$ Sketching Model.} Querying with vector $\u \in \F_2^{n^2}$ returns $\u^{\T} \vec(\A)$ over $\F_2$.
\item {\bf Linear Sketching Model.} Querying with vector $\u \in \R^{n^2}$ returns $\u^{\T} \vec(\A)$ over $\R$.
\end{itemize}

There is a relationship regarding lower bounds for the $\Mv$ model vs.~lower bounds for the $\uMv$ or general linear sketching model. In particular, any query in the $\Mv$ can be simulated by $n$ queries in the linear sketching model (in fact, in the $\uMv$ model by taking $\u$ to be the $n$ standard basis vectors one at a time). We often simply state lower bounds for the $\uMv$ or linear sketching models, but using this relationship, we obtain the entries in \tabref{results} for the $\Mv$ model (with the exception of the \textsc{Find}\BPC problem, where we obtain a stronger lower bound in \secref{findbpc}).

\paragraph{Finding vs.~Detecting.} While we mostly focus on detection problems, we also consider the variant where the algorithm should output the planted clique if there is one. We denote this by adding $\mathsf{Find}$ before the problem name (e.g., for the \textsf{Find}\BPC problem/game, the algorithm/protocol should output the planted $r \times s$ biclique). For many of the models we study, it is straightforward to find the clique by using only a factor of $\mathrm{polylog}(n)$ more queries than for detection. We describe the upper bounds in \secref{upper-bounds}, and we prove a communication lower bound for the \textsf{Find}\BPC game in \secref{findbpc}, which implies a lower bound for matrix-vector queries.

\subsection{Algorithms for Detecting and Finding}
\seclab{upper-bounds}

We review algorithms in the query models listed above. We start with the \PC problem, where $k \geq 10 \log n$ for simplicity. Previous work on the edge-probe model presents a simple sampling algorithm using $O((n/k)^2 \log^2 n)$ queries: choose a subset $B$ of $100(n/k) \log n$ vertices uniformly at random, query all pairs in $B$, and compute the largest clique in this induced subgraph~\cite{racz2019finding}. If there is no planted clique, then the largest induced clique has size at most $3 \log n$ with high probability; otherwise, there is an induced clique $B'$ of size at least $4 \log n$ with high probability. To actually find the clique, the next step is to query all neighbors of $B'$, which reveals the whole planted clique using a total of $O((n/k)^2 \log^2 n + n \log n)$ edge-probe queries. The same general idea leads to algorithms for the \SRPC, \BPC, and \PPC problems as well (for detecting and finding). 

We mention two improvements to the edge-probe algorithm in the $\uMv$ and $\Mv$ models. For both models, the query vectors may have bit-complexity $O(\log n)$ in each entry. In the $\uMv$ model, we can query all pairs in $B$ by using only $O((n/k)^2 \log n)$ queries, saving a $\log n$ factor (use exponentially increasing entries to simulate $O(\log n)$ edge-probe queries with one $\uMv$ query). In the $\Mv$ model, we can query with an indicator vector to receive all neighbors of a vertex. Again by using exponentially increasing entries, we can query $O(\log n)$ vertices at a time. Therefore, we can query all pairs in $B$ with $O(n/k)$ queries; we can also find the planted $k$-clique with an additional $O(k)$ queries by looking at the shared neighborhood of $B'$.

We also note that a single query suffices when $k \geq c \sqrt{n}$ for a large enough constant $c > 1$ in the $\uMv$, $\Mv$, and general linear sketching models. We can use a single query to detect a planted $k$-clique with constant probability by counting the edges (i.e., ones in the matrix). Indeed, the total number of edges is at least $\frac{1}{2}\binom{n}{2} + \binom{k}{2} \geq n^2/2 + c' n$ when there is a planted clique. Otherwise, it is at most $\frac{1}{2}\binom{n}{2} + c'' n$ with constant probability for some $c'' < c'$, allowing us to distinguish the two cases. In light of this, we focus on the case of $k = o(\sqrt{n})$ for the remainder of the paper.

\subsection{Communication Complexity Preliminaries}
\seclab{disj}
We consider a multi-player communication model, where $t \geq 2$ players communicate via a publicly shared blackboard (i.e., all players see all messages). The total number of bits written on the blackboard is the measure of communication. This model generalizes both point-to-point and broadcast models, and hence, our lower bounds hold for both the message passing and broadcast settings. We let $\Pi$ denote the collection of all messages written on the blackboard. Abusing notation slightly, we use $\Pi$ for both the protocol and the transcript $\Pi \in \{0,1\}^*$ in bits. The communication cost is the length of $\Pi$, which we denote as $|\Pi|$, in the worst case over the support of the input distribution.
In other words, we consider 
the \emph{randomized communication complexity} (see, e.g.,~\cite{rao_yehudayoff_2020}). At the termination of the communication protocol, one of the players must output the answer using a function of $\Pi$ with no constraints on the computation time (e.g., in the games defined above, the player should output which of the two distributions the input has been sampled from). We consider the success probability of randomized protocols (players have access to both public, shared random bits and private random bits). Throughout, the exact success probability will not be important, and we consider the randomized communication complexity of solving a problem with constant success probability, e.g., 9/10. 

We use a standard $\Omega(n)$ lower bound on the 2-player \textsc{Unique Disjointness} game~\cite{kalyanasundaram1992probabilistic, razborov1992distributional}.
Two players each have a bitstring $\x, \y \in \{0, 1\}^n$. They are promised that one of the following two cases holds: either (i) for all $i \in [n]$ either $x_i = 0$ or $y_i=0$ or both, or (ii) there is a unique $i \in [n]$ such that $x_i=y_i=1$ and for all $i' \neq i$, either $x_i = 0$ or $y_i=0$ or both. The \textsc{Unique Disjointness} game is to communicate and determine which case they are in.

%% file: PC-COLT.tex
\section{Warm-up: Lower Bound for $\PC$ in the Edge-Probe Model}
\seclab{pc-edge}

We present a simple proof demonstrating the main ideas of our reduction method. We first explain the graph decomposition, and then we use this to prove a communication lower bound for the XOR version of the \PC game in \thmref{xor-pc}. As a consequence, in \corref{xor-pc}, we provide an alternate proof of the R\'{a}cz-Schiffer lower bound of $\Omega(n^2/k^2)$ edge-probe queries for the $\PC$ problem~\cite{racz2019finding}.

\seclab{clique_partition}

The key aspect of our communication lower bounds is using a graph decomposition into a set of edge-disjoint cliques.\footnote{More formally, a set of {\em edge-disjoint} cliques is a collection of subsets of vertices $V_1,\ldots,V_\ell$ such that for all $i \neq j$ the subsets $V_i$ and $V_j$ intersect in at most one vertex.} By ensuring that the cliques are edge-disjoint, while covering most of the graph, we can partition edges among the players while preserving the input distribution. 


\begin{lemma}[Lemma 6.6 in \cite{conlon2012short}]
    \lemlab{clique_partition}
    Let $k \geq 2$ and $n$ be positive integers and let $f(n, k)$ denote the minimum number of
cliques, each on at most $k$ vertices, needed to clique partition the complete graph $K_n$. If
$n > k$, then
$f(n, k) = \Theta\left( \max \left\{(n/k)^2, n \right\} \right) $.
\end{lemma}
\begin{remark}[Number of uncovered edges]\remlab{edges}
    \remlab{partition_clique} First, recall that we are interested in the case when $k = o(\sqrt{n})$. In this regime, the above lemma can be strengthened to show that $f(n,k) = (1+o(1))\frac{n^2}{k(k-1)}$, which is essentially best possible~\cite{conlon2012short}.
    In such a clique partition, there are $\Omega(n^2)$  edges belonging to cliques of size $\Omega(k)$.  Indeed, assume there are $m_1$ cliques of size $\Omega(k)$ and $m_2$ cliques of size $o(k)$, and observe that $m_1 + m_2 = \Theta\left( (n/k)^2 \right)$. If there were only $o(n^2)$ edges belonging to cliques of size $\Omega(k)$, the total number of edges would be $
        o(n^2) + m_2 \cdot o(k^2) \leq o(n^2) + \Theta\left( (n/k)^2 \right) \cdot o(k^2) = o(n^2)$, a contradiction. Thus, $m_1 = \Omega(n^2)/\Theta(k^2) = \Omega\left(n^2/k^2\right).$
\end{remark}

\begin{remark}[Size of cliques] \remlab{size}
	The above lemma only guarantees cliques of size at most $k$. However,  by slightly changing constants, we can guarantee $\Theta(n^2/k^2)$ cliques of size exactly $k$. Indeed, by a standard counting argument, a constant fraction of the cliques must have size at least $\alpha k$ for a constant $\alpha \in (0,1)$. Therefore, we apply lemma with $k' = k/\alpha$, and then find $\Theta(n^2/k^2)$ cliques of size exactly $k$ by restricting to the subcliques of the cliques of size $k'$ if necessary.
\end{remark}

 

\begin{theorem}\thmlab{xor-pc}
Any protocol that solves the XOR version of the $\PC$ game with constant success probability must communicate at least $\Omega(n^2/k^2)$ bits.
\end{theorem}
\begin{proof}
We reduce to the 2-player {\sc Unique Disjointness} game with input length $\ell = \Theta(n^2/k^2)$. Let Alice and Bob have inputs $\x, \y \in \{0,1\}^\ell$, respectively. We use $\x,\y$ to build a random input graph $G$ as follows. 
First, randomly permute the vertex labels. Then, use \lemref{clique_partition} to obtain a collection $S$ of $\Theta(n^2/k^2)$ edge-disjoint cliques with $k$ vertices; for each edge not covered by $S$, choose each of them with probability $1/2$ independently, call this graph $G'$, and give it to Alice (see \remref{edges} and \remref{size} for details about the number of uncovered edges and the clique size, respectively).
Index the subgraphs as $S = \{Z_1,\ldots, Z_\ell\}$. We repeat the following process independently for each $i \in [\ell]$. Alice and Bob will receive graphs $G_1^i$ and $G_2^i$ based on $x_i$ and $y_i$, and these graphs will be supported on the vertices of $Z_i$. Color all edges of a $k$-clique $K_k^i$ with four colors uniformly at random using public randomness. Then,
\begin{itemize}
    \item $\x_i = 0 \implies$ add all edges in $K_k^i$ with colors 1 or 3 to $G_1^i$ 
    \item $\x_i = 1 \implies$ add all edges in $K_k^i$ with colors 1 or 2  to $G_1^i$ 
    \item $\y_i = 0 \implies$ add all edges in $K_k^i$ with colors 1 or 4 to $G_2^i$ 
    \item $\y_i = 1 \implies$ add all edges in $K_k^i$  with colors 3 or 4 to $G_2^i$ 
\end{itemize}  
Define $G = G' \cup \left(\bigcup_{i = 1}^\ell G_1^i \oplus G_2^i\right)$. We claim that if $(x_i,y_i) \neq (1,1)$ for all $i \in [\ell]$, then $G$ is distributed according to $H_0$. Each possible edge is included with probability 1/2 either because of the random coloring or it is in $G'$. Indeed, for each of the three combinations $(0,0), (1,0), (0,1)$, exactly two colors of edges end up in $G_1^i \oplus G_2^i$. Otherwise, if $(x_i,y_i) = (1,1)$ for some $i$, then all four colors of edges appear in $G_1^i \oplus G_2^i$, and hence this is the planted clique. By randomly permuting the vertices at the beginning (with public randomness), each $k$-clique is equally likely.  A protocol solving the XOR version of the $\PC$ game also solves {\sc Unique Disjointness} on $\ell$ bits and must communicate $\Omega(\ell) = \Omega(n^2/k^2)$ bits.
\end{proof}

 Let $G_1, G_2$ denote the adjacency matrices for Alice and Bob, respectively. Assume there is a $q$ query algorithm in the $\F_2$ sketching model that solves the $\PC$ problem with constant probability. This can be implemented by having Alice compute her $q$ sketches on $G_1$ and then she sends these $q$ bits to Bob. Then, Bob can complete the execution of the algorithm on $G_1 \oplus G_2$ locally and solve the XOR version of the \PC game. Thus, $q = \Omega(n^2/k^2)$ by \thmref{xor-pc}, and we get the following.
\begin{corollary}\corlab{xor-pc}
For the $\PC$ problem, $\Omega(n^2/k^2)$ queries in the $\F_2$ sketching model are required to distinguish $H_0$ and $H_1$ with constant probability.
\end{corollary}

 
\section{Parameter Estimation Game and Multi-player Communication}
\seclab{IC}

We next prove communication lower bounds that imply query lower bounds for the general linear sketching model. While similar results have appeared before (e.g.,~\cite{bar2004information, braverman2016communication, weinstein2015simultaneous}), we are unaware of any results that suffice for the distributional lower bounds that we need for our reductions. 
The entropy of $X$ is  $H(X) = -\sum_x p_x \log_2 p_x$. 
The mutual information is 
$
    I(X; Y) = H(X) - H(X|Y) = H(Y) - H(Y|X) = I(Y; X).
$

\begin{definition}[Hellinger Distance]
	Consider two probability distributions $f,g: \Omega \to \mathbb{R}$. The square of the Hellinger distance between $f$ and $g$ is
	$
		h^2(f,g) := \frac{1}{2} \int_{\Omega} \left( \sqrt{f(x)} - \sqrt{g(x)} \right)^2\ dx.
	$
\end{definition} 


We consider a version of the multi-party \textsc{Unique Disjointness} game, where $n$ players each receive an $m$-dimensional binary vector, and they determine whether there is some coordinate such that every player's vector has a one in this coordinate. We also define an input distribution where the vectors are either uniformly random, or there is a planted coordinate that is all ones.

\medskip \noindent {\bf Parameter Estimation ($\PE$) game.} Let $B \in \{0, 1\}$ be a binary variable, and let $\V \in \{0, 1\}^m$ be a random binary vector (the distribution of $\V$ will depend on $B$).  When $B = 0$, then $\V$ is the all zeros vector. When $B = 1$, then there is exactly one entry of $\V$ equal to 1, and the entry is chosen uniformly at random. We define two distributions: 
$\mu_0 = \texttt{Bernoulli}(1/2)$ and $\texttt{Bernoulli}(1).$
Now suppose $\V = \v$, and the $n$ players each obtain a vector $\X^{(i)} \in \{0, 1\}^m$, where $X^{(i)}_j \sim \mu_{v_j}$. They need to communicate with each other to determine the value of $B$ with error probability at most $\delta$. We assume that $m$ and $n$ are comparable, i.e., $m = O( n^\alpha)$ for some constant $\alpha > 0$. 
To set notation, let $\mathcal{Z} \subset \left(\mathscr{X}^{(1)}\right)^m \times \left(\mathscr{X}^{(2)}\right)^m \times \cdots \times \left(\mathscr{X}^{(n)}\right)^m$ be the set of inputs.
The $\PE$ game corresponds to computing $f : \mathcal{Z} \to \{0, 1\}$, which outputs $B$ on inputs $\X^{(1)}, \ldots, \X^{(n)}$, where the inputs are drawn from the distribution described above (depending on $B$).
For convenience, let $\X = (\X^{(1)}, \X^{(2)}, \ldots, \X^{(n)})$, and let $\X_j = (X^{(1)}_j, X^{(2)}_j, \ldots, X^{(n)}_j)$. Also, let $\Pi \in \{0, 1\}^*$ be a randomized protocol, where $\Pi(\X)$ is the transcript when the players have 
$\X$ as inputs, and $|\Pi(\X)|$ denotes its length in bits. 
We consider a function $g : \{0, 1\}^* \to \{0, 1\}$, such that when $B = 0$, then $g(\Pi(\X)) = 0$ with probability at least $1 - \delta$, and when $B = 1$, then $g(\Pi(\X)) = 1$ with probability at least $1 - \delta$, where the randomness is from both the input $\X$ and the protocol $\Pi$, i.e., the players have shared public and also private randomness. 
In other words, $g$ is the estimator for the parameter estimation problem. We provide a lower bound on the information and communication complexity of solving the \PE game, which will be the basis of several of our results. 


Next, we recall a standard communication lower bound (see, e.g., \cite{rao_yehudayoff_2020}). 

\begin{proposition}
\proplab{CC_lb}
For a protocol $\Pi$ and distribution $\mu$ of inputs, 
$\max_{\X' \in \mathsf{supp}(\mu)} |\Pi(\X')| \geq I(\X;\Pi).$
%
%
\end{proposition}

\subsection{Direct Sum and Communication Lower Bound}

For distributions $\mu_0, \mu_1$ over the same sample space, we write $\mu_1 \leq c \cdot \mu_0$ if the point-wise density of $\mu_0$ is at most $c$ times larger than $\mu_1$ for $c > 0$. For $\mu_0, \mu_1$ defined above, we have $c=2$ and that only a one-sided guarantee is possible (as $\mu_1$ has no mass on 0). We use the distributed strong data processing inequality (Distributed SDPI).
Let $\beta(\mu_0,\mu_1)$ denote the {\em SDPI} constant, which is the infimum over real $\beta\geq 0$ such that $I(B;\Pi) \leq \beta \cdot I(\X;\Pi)$ where $B \rightarrow \X \rightarrow \Pi$ forms a Markov chain. This inequality holds with $\beta=1$, which is the {\em data processing inequality}. For our results, it suffices to take $\beta=1$, but for completeness, we state the stronger version of the following theorem.

\begin{theorem}[Theorem 3.1 in~\cite{braverman2016communication}]\thmlab{ddpi}
	Suppose $\mu_1 \leq c \cdot \mu_0$ and $\beta(\mu_0,\mu_1) = \beta$. Then, 
	$
		c'(c+1)\beta \cdot I(X;\Pi \mid B=0) \geq h^2(\Pi|_{B=0}, \Pi|_{B=1}), 
	$
	where $c' > 0$ is an absolute constant. The same holds conditioned on $B=1$ instead of $B=0$.
\end{theorem}


The challenge is to lower bound the information, conditioning on the distribution when $B=0$, which is the utility of the above theorem. When the protocol is correct with constant probability, the Hellinger distance is also a constant (via a standard connection with total variation distance), and when $\beta = \Theta(1)$, then \thmref{ddpi} provides an $\Omega(1)$ lower bound on the information. This suffices for our purposes because we use a direct sum over many instances and only need an $\Omega(1)$ lower bound on the information to achieve the communication lower bound. 
We can decompose the $\PE$ game on $m$ coordinates to  a single coordinate, which follows from standard properties of mutual information, such as subadditivity, since conditioned on $B = 0$, all $\X_1,\ldots,\X_m$ are independent (see e.g.~\cite{bar2004information, braverman2016communication}). 
\begin{lemma}
\lemlab{lem_direct_sum}
    Fix the input distribution of $\X$ when $B = 0$. 
    Then, \[I(\X; \Pi | B = 0) \geq \sum_{j=1}^m I( \X_j; \Pi | B = 0).\]
\end{lemma}  
\begin{proof}
     By definition, $I(\X; \Pi | B = 0) = H(\X | B = 0) - H(\X| \Pi, B = 0)$. Observe that we have $H(\X | B = 0) = \sum_{j=1}^m H(\X_j | B = 0)$ since given $B = 0$, all $\X_1,\ldots,\X_m$ are independent. Also, by subadditivity, 
     $H(\X| \Pi, B = 0) \leq \sum_{j=1}^m H(\X_j | \Pi, B = 0).$
     Putting these together, we have that
     $
         I(\X; \Pi | B = 0) \geq \sum_{j=1}^m H(\X_j | B = 0) - \sum_{j=1}^m H(\X_j | \Pi, B = 0) = \sum_{j=1}^m I( \X_j; \Pi | B = 0).
     $
\end{proof}
\begin{theorem}
\thmlab{thm_ic_game}
    Assume that $m = \mathrm{poly}(n)$. The communication complexity of the $n$-player $\PE$ game on $m$ coordinates with error probability $\delta$ is $\Omega(m)$ assuming that $\delta \leq 1/10$.
\end{theorem}
\begin{proof}
 The single coordinate $\PE$ game is that for a specific $j \in [m]$ and $V_j = v_j \in \{0, 1\}$, each of the~$n$ players receives an instance of the variable $X^{(i)}_j \sim \mu_{v_j}$ for $i \in [n]$. Their task is to communicate with each other to determine the value of $V_j$ with error probability at most $\delta$. 
Let $\Pi$ be a randomized protocol that solves the $n$-player $\PE$ game on $m$ coordinates with error probability~$\delta$. Our goal is to show that  
\begin{equation}\eqlab{info-lb-PE}
I( \X_j; \Pi | B = 0) = \Omega(1) \quad \mathrm{for\ all}\ j \in [m].
\end{equation}
By \propref{CC_lb} and \lemref{lem_direct_sum}, we lower bound the communication by
    $$
        I(\X; \Pi | B = 0) \geq \sum_{j=1}^m I( \X_j; \Pi | B = 0) = m\cdot\Omega(1) = \Omega(m).
    $$
To show \Eqref{info-lb-PE}, we use \thmref{ddpi}. We consider the single coordinate $\PE$ game on coordinate~$j$. We construct a protocol $\Pi'(\X_j)$ to solve the single coordinate $\PE$ game. Our method is to construct another random matrix $\X'$ as follows. Using public randomness, players choose a uniformly random $j' \in [m]$, and let $\X'_{j'} = \X_j$. For $\ell \neq j'$, let $\X'_{\ell}$ be a random vector where each entry is an independent $\texttt{Bernoulli}(1/2)$ variable, sampled by each player independently using private randomness. Since $m = \mathrm{poly}(n)$, the probability that any $\X'_{\ell}$ is an all ones vector is exponentially small. Then, let $\Pi'(\X_j) = \Pi(\X')$.
By this construction, when $B = V_j$, we have that $\Pi$ has the same distribution as $\Pi'(\X_j)$. Since $\Pi$ could determine the value of $B$ with error probability~$\delta$, $\Pi'(\X_j)$ can also determine the value of $V_j$ with error~$\delta$.
Thus, by \thmref{ddpi},
$
    I( \X_j; \Pi | B = 0) = I( \X_j; \Pi'(\X_j) | V_j = 0) = \Omega(1)\ \mathrm{for\ all}\ j \in [m],
$
where we use $\beta=1$ and $c=2$ and that the squared Hellinger distance is $\Theta(1)$ since the success probability is a constant.
\end{proof}

\section{Planted Clique Lower Bound for Linear Sketching}
\seclab{sec_lower_bound_PC}\seclab{pc}


\begin{theorem}
    \thmlab{thm_lower_bound_PC}
    For $k = n^{\gamma}$ where $0 < \gamma < \frac{1}{2}$, any protocol with $\Theta(k^2)$ players that solves the $\PC$ game with constant success probability must communicate $\Omega\left(n^2/k^2\right)$ bits.
\end{theorem}

\begin{proof}
We reduce from the \PE game (\secref{IC}) to the \PC game to get hardness of \PC from hardness of \PE. 
Given a complete graph with $n$ vertices, by \lemref{clique_partition}, we can partition most of the edges (or equivalently, vertex pairs) by $\Theta\left(n^2/k^2\right)$ cliques of size $k$. We consider the  $\PE$ game with $\binom{k}{2}$ players, each having inputs with $\Theta\left(n^2/k^2\right)$ coordinates (i.e., each player is responsible for one edge in each potential clique). Each $V_j$ corresponds to a clique of size $k$, and the indicator vector for its $\Theta(k^2)$ edges corresponds to the binary vector~$\X_j$ (using an arbitrary indexing of the edges). The uncovered edges can be sampled with a public coin to appear with probability $1/2$, and they can be given to any player without loss of generality (see \remref{edges} and \remref{size} for details about the number of uncovered edges and the clique size, respectively). Using public randomness, the players randomly relabel all vertices, so that the location of the planted clique is random). By this construction, we have the $\PE$ to $\PC$ translation:
$ B = 0 \mbox{ corresponds to } G(n, 1/2)$ and $B = 1 \mbox{ corresponds to } G(n,  1/2, k).$
Hence, the $\Theta(k^2)$ players can solve the $\PE$ game by detecting the planted clique. The randomized communication complexity of the $\PC$ game is $\Omega\left(n^2/k^2\right)$.
\end{proof}


    The communication lower bound of the $\PC$ game with $\Theta(k^2)$ players is $\Omega(n^2/k^2)$. A single query in the general linear sketching model can be simulated with $O(k^2 \log n)$ bits of communication since there are $\Theta(k^2)$ players. Thus, any algorithm that solves the \PC problem with constant success probability must use $\widetilde \Omega(n^2/k^4)$ queries, and we get the following.

\begin{corollary}
\corlab{coro_PC_linear_sketching}
    Let $k = n^{\gamma}$ where $0 < \gamma < 1/2$. Then, $\widetilde \Omega(n^2/k^4)$ general linear sketching queries are necessary to solve the $\PC$ problem with constant success probability.
\end{corollary}

%% file: BPC.tex
\seclab{bpc}

For the $\BPC$ problem, we first state the bipartite version of the graph decomposition lemma.

\begin{lemma}
\lemlab{biclique_partition}
Given $n, r, $ and $s$, we can use $\lceil n/r\rceil\cdot \lceil n/s\rceil$ edge-disjoint bicliques to cover an $n\times n$ complete bipartite graph. Moreover, $\lfloor n/r\rfloor\cdot \lfloor n/s\rfloor$ of them are of size $r\times s$.
\end{lemma}
\begin{proof}
Let $a = \lceil n/r\rceil$ and $b = \lceil n/s\rceil$. We can partition the vertices on the left side into $a$ sets $U_1, \dots, U_a$ so that $|U_1| = \cdots = |U_{a-1}| = r$ and $|U_a| = n - (a-1)r$. Also we can partition the vertices on the right side into $b$ sets $V_1, \dots, V_b$ so that $|V_1| = \cdots = |V_{b-1}| = s$ and $|V_b| = n - (b-1)s$. The $a\cdot b$ bicliques formed by $U_i$ and $V_j$ for all $i,j$ can cover the whole bipartite graph.
\end{proof}

Now we can prove a lower bound using the same strategy as \thmref{thm_lower_bound_PC}. 

\begin{theorem}
For the $\BPC$ problem, suppose $rs\le n$. Then $\widetilde {\Omega}(n^2/(rs)^2)$ general linear sketching queries are required to distinguish $H_0$ and $H_1$ with constant probability.
\end{theorem}

\begin{proof}
We use \lemref{biclique_partition} to randomly partition a complete bipartite graph. Then we consider the $\PE$ game for $rs$ players, where each player receives $\lfloor n/r\rfloor\cdot \lfloor n/s\rfloor = \Theta(n^2/(rs))$ coordinates. 
We reindex the edges (regardless if they exist or not) in each biclique from $1$ to $rs$. The $i$-th edge is present in the $j$-th biclique if and only if the value player $i$ holds at coordinate $j$ is one. For the negligible amount of edges that are not covered by these bicliques (e.g., the subgraphs with size other than $r\times s$), we let the graph contain each of them with probability $1/2$ using public randomness. 
The communication lower bound of the $\PE$ game with $\Theta(rs)$ players on $\Theta(n^2/(rs))$ coordinates is $\Omega(n^2/(rs))$, which implies the same lower bound for the $\BPC$ game. A single query in the general linear sketching model can be simulated with $O(rs \log n)$ bits of communication since there are $\Theta(rs)$ players. Thus, any algorithm that solves the \BPC problem with constant success probability must use $\widetilde \Omega(n^2/(r^2s^2))$ queries.
\end{proof}

We also design an algorithm when $r$ is larger than $C\sqrt{n\log n}$ (w.l.o.g. we suppose $r>s$) for some constant $C$, to further close the gap between the upper bound and the lower bound. 
\begin{theorem}
For the $\BPC$ problem, suppose $r\ge C\sqrt{n\log n}$ for $C > 16$. Then there exists an algorithm which can distinguish $H_0$ and $H_1$ with high probability, using $\widetilde{O}(n^2/(r^2s))$ $\uMv$ queries. 
\end{theorem}
\begin{proof}
The algorithm is as follows. We first randomly sample $\widetilde{O}(n/s)$ columns. If there is a planted biclique, then with high probability at least one column that belongs to the planted column will be sampled. We then partition the columns into groups with size $r^2/(16n\log n)$ each, and compute $\sum_{i\in S}x_i$ for each group $S$ where $x_i$ is the sum of column $i$. The algorithm returns $1$ (i.e., there is a planted biclique) if there is an $S$ such that $\sum_{i\in S}x_i \ge n/2\cdot |S| + r/4$. And it returns $0$ if there is no such $S$. Since the sum of a group can be computed using one $\uMv$ query, our algorithm will use $\widetilde{O}(n^2/(r^2s))$ $\uMv$ queries in total.

We now prove that our algorithm succeeds with high probability. We consider a set $S$ of $r^2/(16n\log n)$ columns. If $S$ does not contain any columns in the planted biclique, by Hoeffding's inequality, we have
$\Pr\left(\sum_{i\in S} x_i - n/2 \cdot |S|\ge r/4\right)\le \exp(-\frac{2r^2/4^2}{n \cdot |S|}) = e^{-2\log n} = 1/n^2$, where we consider $\sum_{i\in S} x_i$ as the sum of $n|S|$ random variables with $|S| = r^2/(16n\log n)$. Therefore, $\sum_{i\in S} x_i$ is less than $n/2 \cdot |S| + r/4$ with probability at least $1-1/n^2$.

Now consider the case that $S$ contains at least one planted column. Let $U$ denote the planted biclique. Since the expected value of $x_i$ is $n/2 + r/2$ for $i\in U$, Hoeffding's inequality implies that 
$$\Pr\left(\left|\sum_{i\in S} x_i - n/2\cdot |S| - r/2\cdot |S\cap U|\right|\le -r/4\right)\le \exp\left(-\frac{2r^2/4^2}{n \cdot |S| - r\cdot |S\cap U|}\right) \le 1 / n^2.$$ 
Thus, $\sum_{i\in S} x_i$ will be greater than $n/2 \cdot |S| + r/2 \cdot |S\cap U| - r/4 \ge n/2 \cdot |S| + r/4$ with probability at least $1 - 1/ n^2$.
A union bound implies that if there is not a planted clique, with high probability the sum of each group will be smaller than $n/2 \cdot |S| + r/4$ simultaneously. Otherwise with high probability we can find a group whose sum is larger than $n/2\cdot |S| + r/4$. Thus our algorithm can output the correct answer with high probability.
\end{proof}

Considering lower bounds on the the query complexity in the $\F_2$ sketching model (and hence the edge probe model), we obtain a better lower bound for the XOR version of the $\BPC$ game.

\begin{theorem}
Any protocol solving the XOR version of the $\BPC$ game with constant success probability must communicate $\Omega(n^2/(rs))$ bits, and hence, $\Omega(n^2/(rs))$  queries are required to solve the \BPC problem in the $\F_2$ sketching model.
\end{theorem}
\begin{proof}
We again use \lemref{biclique_partition} to randomly partition a complete bipartite graph. Now we consider the two-player {\sc Unique Disjointness} game with input length $\lfloor n/r\rfloor\cdot \lfloor n/s\rfloor$. For each edge in an $r\times s$ biclique we uniformly randomly assign a color among 4 colors. Then we consider the bitstrings Alice and Bob hold, and construct graphs $G_1$ and $G_2$ as follows:
\begin{itemize}
    \item If Alice has a 0, add edges with color 1 or 3 in the corresponding biclique to $G_1$.
    \item If Alice has a 1, add edges with color 1 or 2 in the corresponding biclique to $G_1$.
    
    \item If Bob has a 0, add edges with color 1 or 4 in the corresponding biclique to $G_2$.
    \item If Bob has a 1, add edges with color 3 or 4 in the corresponding biclique to $G_2$.
\end{itemize}
Finally, we construct graph $G = G_1 \oplus G_2$, namely, an edge occurs in $G$ if and only if it occurs in exactly one of $G_1$ and $G_2$. It can be verified that if Alice and Bob both have a 1 on the same position, $G$ will contain the corresponding biclique. Otherwise $G$ will randomly contain each edge with probability $1/2$. Therefore, we finish the reduction and obtain an $\Omega(n^2/rs)$ lower bound.
\end{proof}

%% file: FindBPC.tex
\seclab{findbpc}

We consider the \textsf{Find}\BPC game, where there may be a planted $r \times s$ biclique in an $n \times n$ bipartite graph, and the goal is to output all vertices of the biclique if it exists. Considering  the case when $r = s = k$, the algorithm in \secref{upper-bounds} uses $\widetilde O(n/k)$ queries to solve the \textsf{Find}\BPC problem in the $\Mv$ model. We provide a nearly-matching lower bound, showing that $\widetilde \Omega(n/k)$ queries are necessary. In fact, we can use a similar strategy to obtain both a communication lower bound for the \textsf{Find}\BPC game and a query complexity lower bound for the \textsf{Find}\BPC problem in the $\Mv$ model. 


\begin{theorem}\thmlab{find-bpc-comm}
Let $r$ and $s$ be parameters that satisfy $3\log n \leq r \leq s \leq n/2$. 
\begin{itemize} \item Any $r$-player protocol that solves the \textsf{Find}\BPC game with constant success probability must communicate $\Omega\left(n \right)$ bits. 
\item Any algorithm
that solves the \textsf{Find}\BPC problem with constant success probability must use \\  $\Omega\left(n/(r \log n)\right)$ queries in the $\Mv$ model.
\end{itemize}
\end{theorem}
\begin{proof}
For the first part of the theorem, we reduce the \textsf{Find}\BPC game to a ``promise'' variant of the \PE game defined in \secref{IC}. We refer to this variant as \textsf{Find}\PE, where the parameter $V$ is always set to one, but the players must output the index of the coordinate that is all ones (which is promised to exist when $V=1$). For consistency with the \textsf{Find}\BPC formulation, we let $r$ denote the number of players and $n$ denote the length of the vectors given to each player. We assume that $r \geq 3 \log n$ so that with high probability the only all ones coordinate is the planted coordinate. 

First, note that the communication lower bound for the original \PE game implies a lower bound for  \textsf{Find}\PE. To see this, we show how a \textsf{Find}\PE protocol can solve the \PE game with a negligible increase in communication. Given a \PE instance, the players run the \textsf{Find}\PE algorithm (even in the case of $V=0$, as long as the protocol aborts if it uses more communication than it would on a $V=1$ instance). If it outputs the index of a column, the players can sample $O(1)$ bits from this column to determine if it is all ones or random. If the \textsf{Find}\PE algorithm outputs anything else, then we know that $V=0$. This requires only $O(1)$ extra bits to succeed with constant probability, implying \textsf{Find}\PE requires $\Omega(n)$ bits of communication.

Now we explain the connection to \textsf{Find}\BPC. We begin by constructing a random $n \times n$ matrix. We choose $n-r$ rows uniformly at random, and we independently sample the entries of these rows from $\mathtt{Bernoulli}(1/2)$. For the remaining $r$ rows, we assign one row to each of the $r$ players. Among the $n$ entries of the rows, we choose $s-1$ at random, and set all entries to be one in each of these chosen rows (e.g., we plant $s-1$ all ones columns in the $r \times n$ submatrix). 
Overall, we have defined the whole matrix except for $n-s+1$ entries in each of the $r$ rows. By using public randomness, we can assume that the $n^2 - r(n-s+1)$ entries are known to all players.

We embed an instance of the $r$-player \textsf{Find}\PE game in the unset entries, where each player has an input vector of size $n-s+1$. Since \textsf{Find}\PE is a promise variant, we are guaranteed that one of the coordinates is one in all $r$ vectors. In particular, the full $n\times n$ matrix corresponds to the adjacency matrix of a bipartite graph with an $r\times s$ planted biclique (e.g., $r \times s$ all ones submatrix). Using this construction, the players can then execute a protocol for \textsf{Find}\BPC. By doing so, they reveal the location of the all ones coordinate from the \textsf{Find}\PE instance. Therefore, since $s \leq n/2$, we have that $n-s+1 = \Omega(n)$, and the players must communicate $\Omega(n)$ bits, which provides the desired lower bound for solving the \textsf{Find}\BPC game.

Moving on to the second part of the theorem, we can also use the same construction to prove a lower bound on the query complexity in the $\Mv$ model. The $r$ players build the $n\times n$ matrix in the same way as before, and the connection to the \textsf{Find}\PE game is also the same. The difference is that they will now use a protocol for \textsf{Find}\BPC that we derive from a query algorithm in the $\Mv$ model. Recall that the inputs of the $r$ players correspond to an $r \times n$ submatrix (and the rest of the matrix is known to all the players). Therefore, each query in the $\Mv$ model can be simulated by communicating $O(r \log n)$ bits because the players simply need to evaluate the matrix-vector product on the $r \times n$ submatrix (each player handles one row). If the query algorithm uses $q$ queries to solve the \textsf{Find}\BPC problem, then this gives rise to a protocol for this construction that solves the \textsf{Find}\BPC game (and hence the \textsf{Find}\PE game) by communicating $q \cdot O(r\log n)$ bits. Thus, $q = \Omega(n/(r \log n))$ queries are needed to solve the \textsf{Find}\BPC problem in the $\Mv$ model.
\end{proof}

%% file: SRPC.tex
\seclab{srpc}
For the $\SRPC$ problem, since the adversary only removes edges outside the planted clique, we can use the existing edge-probe upper bound (Theorem 1, \cite{racz2019finding}) to obtain the following: suppose $k\ge (2+\eps) \log n$ for some constant $\eps > 0$, then there exists an algorithm which can distinguish $H_0$ and $H_1$  using $\widetilde{O}(n^2/k^2)$ edge-probe queries (see \secref{upper-bounds}; the algorithm is the same as the standard planted clique problem).
We also provide a nearly matching lower bound for the corresponding communication game. The key observation is that we can reduce to the 2-player {\sc Unique Disjointness} game (instead of $k^2$ players) because we now have more flexibility to remove edges that are not in the planted clique.

\begin{theorem} Let $k= n^\gamma$ for any $\gamma \in (0,1/2)$.  Any 2-player protocol that solves the \SRPC game with constant success probability needs to communicate $\Omega(n^2/k^2)$ bits. Hence, $\Omega(n^2/(k^2 \log n))$ queries are required in the general linear sketching model to solve the \SRPC problem.
\end{theorem}
\begin{proof}
By \lemref{clique_partition}, we can partition the complete graph $K_n$ into $\Theta(n^2/k^2)$ edge-disjoint cliques each with size $\Theta(k)$. Then we randomly color each edge in these cliques with red or blue with equal probability. Now we consider the 2-player {\sc Unique Disjointness} game with input strings of  length $\Theta(n^2/k^2)$. Construct a graph $G$ as follows: for each bit, let $G$ contain the red edges in the clique if Alice has a 1, and let $G$ contain the blue edges in the clique if Bob has a 1. For those edges outside the cliques, each of them occurs with probability $1/2$. If there is at most a single 1 in every position, then $G$ can be viewed as an instance under $H_0$. If there is a unique position such that both Alice and Bob have a 1, then $G$ can be viewed as an instance under $H_1$. Thus we reduce the 2-player set {\sc Unique Disjointness} game (\secref{disj}) to our semi-random planted clique, and therefore we get an $\Omega(n^2/k^2)$ lower bound for the \SRPC game. For the query lower bound, we can simulate a single linear sketch query with $O(\log n)$ bits of communication, which implies that any algorithm succeeding with constant success probability must use $\Omega(n^2/(k^2 \log n))$ queries.
\end{proof}

%% file: PPC.tex
 \seclab{ppc}
 
Recall that the $\PPC$ problem is a promise variant of the planted clique problem. Here, there is a set $S$ of $\Theta(n^2/k^2)$ possible subsets of vertices that may contain the clique. This information is known beforehand, and the goal is determine whether the graph is random or a $k$-clique has been planted in one of the subgraphs in $S$. We provide nearly matching upper and lower bounds.

\begin{theorem}
If $|S|=\Theta(n^2/k^2)$ for the $\PPC$ problem, then $\wt \Theta(n^2/k^4)$ queries are sufficient and necessary for constant success probability in the general linear sketching and $\uMv$ models.
\end{theorem}
\begin{proof}
For the lower bound, observe that the proof of \thmref{thm_lower_bound_PC} already uses a set $S$ of $\Theta(n^2/k^2)$ possible subsets of vertices that may contain the clique (via the graph decomposition result \lemref{clique_partition} and \remref{partition_clique}). Hence, we can prove a lower bound using the $\Theta(k^2)$-player version of the $\PE$ game as before, which requires $\Omega(n^2/k^2)$ bits of communication. We can simulate each linear sketching or $\uMv$ query with $O(k^2 \log n)$ bits since there are $\Theta(k^2)$ players. Thus, $\Omega(n^2/(k^4 \log n))$ queries are necessary.

For the upper bound, we sketch the idea of a randomized algorithm using the knowledge of $S$. First, randomly choose a subset $S'$ of $\frac{k^2}{81\log n} $ subgraphs uniformly from $S$. The subgraphs in $S'$ contain a total of $m = \frac{k^2}{81\log n}  \cdot {k\choose 2} =  \Theta(k^4 /\log n)$ edges. We let $a_i$ be $1$ if the $i$-th edge exists and be $0$ if it does not exist. Then, if $S'$ does not contain the planted clique, by Hoeffding's inequality, we have 
$$\Pr\left(\sum_{i = 1}^m a_i - m/2\ge k^2/9\right)\le \exp(-\frac{k^4}{81m}) \le 1/n^2.$$ If $S'$ contains the planted clique, we have 
$$\Pr\left(\sum_{i = 1}^ma_i - {k\choose 2} - (m - {k\choose 2})/2\le -k^2/9\right)\le \exp\left(-\frac{k^4}{81(m - {k\choose 2})}\right) \le 1/n^2.$$ Note that ${k\choose 2} + (m - {k\choose 2})/2 - k^2/9 > m/2 + k^2/9$. As a consequence, we can identify whether there is a planted clique in $S'$ by counting the edges in $S'$ using one $\uMv$ query. Repeating this process $\Theta(n^2/k^2)/\Theta(k^2/\log n) = \widetilde{\Theta}(n^2/k^4)$ times, we have that a union bound ensures that one of the subsets $S' \subseteq S$ contains the planted clique with high probability if there is one in the graph.
\end{proof}

%% file: HH.tex
\seclab{hh}
Our techniques give an $\Omega(n^2/k^2)$ communication lower bound for a corresponding $k^2$-player game, which implies a lower bound of $\widetilde{\Omega}(n^2/k^4)$ general linear sketching queries for the \HH problem~\cite{pmlr-v65-kannan17a}. 

\begin{theorem} 
\thmlab{HH}
Suppose $\sigma_1 \le \sigma_0 \le c\sigma_1$ for some constant $c > 0$. Any algorithm that solves the \HH game with constant success probability requires $\widetilde{\Omega}(n^2/k^4)$ queries in the general linear sketching and $\uMv$ models.
\end{theorem}
\begin{proof}
We sketch the slight modification of the proof of \thmref{thm_lower_bound_PC} for the \HH problem (which in turn needs a slight modification of the proof of \thmref{thm_ic_game}).
First, we let $w = \lfloor n / k\rfloor$, and we randomly choose $w$ disjoint subsets $R_1,\ldots, R_w$ of $k$ rows each (discard the remaining rows if $k$ does not divide $n$). Then, for row $r$ in set $R_i$, we again randomly choose $w$ disjoint sets of $k$ entries $T_{r1}, \ldots, T_{rw}$ and let $U_{ij} = \bigcup_{r\in R_i} T_{rj}$. In particular, we have that $|U_{ij}| = k^2$. 

Now we consider a $k^2$-player communication game that is a modified ``Gaussian version'' of the \PE game from \secref{IC}, where the players have inputs of size $w^2$ each. The modification is that instead of binary variables, we consider the distributions $\mu'_0 = \mathcal{N}(0, \sigma_0^2)$ and $\mu'_1 = \mathcal{N}(0, \sigma_1^2)$, where $\sigma_0$ and $\sigma_1$ are the parameters of the \HH problem. 
In this way, we construct matrix $\A$ so that the entries in $U_{ij}$ for $i,j \in [k]$ are the values that the $k^2$ players hold. For the entries not in any $U_{ij}$ we generate them according to $\mu'_0$. Thus we reduce to the \HH game from this $k^2$-player communication game.
The key step of the proof of \thmref{thm_ic_game} provides a lower bound of $\Omega(1)$ for the information complexity  for each coordinate. We can prove the same bound, again by \thmref{ddpi}. We continue to use $\beta =1$, e.g., the standard data processing inequality. 

The coordinates in the modified game are drawn from the distributions $\mu'_0$ and $\mu'_1$.
Then, we let $f_0(x) = \frac{1}{\sqrt{2\pi}\sigma_0}\exp(-\frac{x^2}{2\sigma_0^2})$ and $f_1(x) = \frac{1}{\sqrt{2\pi}\sigma_1}\exp(-\frac{x^2}{2\sigma_1^2})$ be the probability density function of $\mu_0'$ and $\mu_1'$, respectively. We have that $$\frac{f_1(x)}{f_0(x)} = \frac{\sigma_0}{\sigma_1}\cdot\exp\left(-\frac{x^2}{2}\left(\frac{1}{\sigma_1^2} - \frac{1}{\sigma_0^2}\right)\right).$$ Note that  $\frac{\sigma_0}{\sigma_1} \leq c$, and the exponential term is less than $1$ since $\sigma_1 \le \sigma_0$, and so $\mu'_1\le c\mu'_0$. Thus, the direct sum argument holds (\lemref{lem_direct_sum}), and we also have an $\Omega(1)$ lower bound on the information complexity for each coordinate. Hence, the overall proof strategy implies a lower bound of $\Omega(w^2) = \Omega(n^2/k^2)$ for the $k^2$-player game.

For the query lower bound, we note that each general linear sketching or $\uMv$ query can be simulated with $O(k^2 \log n)$ bits of communication as there are $k^2$ players. Therefore, this shows that $\wt \Omega(n^2/k^4)$ queries are necessary to solve the \HH problem.
\end{proof}

The algorithm from~\cite{pmlr-v65-kannan17a} can be simulated in the query models that we consider. The main idea is to randomly sample $\wt \Theta(n/k)$ entries in each row. Then, using these, it is known how to distinguish the two hypotheses from the \HH problem as long as $\sigma^2_1 > 2\sigma^2_0$. Hence, in the edge-probe model, we can sample the entries with $\wt \Theta(n^2/k)$ queries. The same upper bound trivially holds for the linear sketching and $\uMv$ models. In the $\Mv$ model, we can sample with $\wt \Theta(n/k)$ queries. 

For the \HH problem, we leave open the question of tightening the bounds. We also note that there is a bound of $\wt \Theta(n^2/k^2)$ in the statistical query model, depending on $\sigma_0$ and $\sigma_1$~\cite{pmlr-v65-kannan17a}.

%% file: SPCA.tex
\seclab{spca}









We consider the sub-Gaussian version of the \SPCA problem, using parts of a known reduction from the \PC problem~\cite{berthet2013complexity}.
The empirical variance of $t$ vectors $\X_1,\ldots,\X_t \in \mathbb{R}^d$ in direction $\v$ is 
$$
\emvar(\v) = \frac{1}{t} \sum_{i=1}^t (\v^\top \X_i)^2.
$$

Let $\theta$ and $k$ be parameters, and let $\zeta \in (0,1)$ be a fixed, small constant. 
We let $\mathcal{D}_0$ denote the set of product distributions over $t$ i.i.d.~vectors $\X_1,\ldots,\X_t \in \mathbb{R}^d$ such that for all unit vectors $\v$, 
\begin{equation}\eqlab{spca-d0}
\Pr\left[|\emvar(\v) - 1 | > 4\sqrt{\frac{\log(2/\zeta)}{d}} + 4\frac{\log(2/\zeta)}{d}\right] \leq \zeta.
\end{equation}
In other words, $\mathcal{D}_0 := \{\mathbf{P}_0 \mid \Eqref{spca-d0} \mbox{ holds}\}$, and $\mathcal{D}_0$ contains, e.g., isotropic distributions.

We let $\mathcal{D}_1^{k,\theta}$ denote the set of product distributions over $t$ i.i.d.~vectors $\X_1,\ldots,\X_t \in \mathbb{R}^d$ such that for all unit vectors $\v$  with at most $k$ nonzero entries ($\|\v\|_0 \leq k$), we have
\begin{equation}\eqlab{spca-d1}
\Pr\left[\left(\emvar(\v) - (1+\theta)\right) < -2\sqrt{\frac{\theta k \log(2/\zeta)}{d}} - 4\frac{\log(2/\zeta)}{d}\right] \leq \zeta.
\end{equation}
Similarly, $\mathcal{D}_1^{k,\theta} := \{\mathbf{P}_1 \mid \Eqref{spca-d1} \mbox{ holds}\}$.

\paragraph{\SCDC problem.} Define two hypotheses to test; the inputs $\X_1,\ldots,\X_t$ are drawn from $\mathbf{P}$ such that
$$
\mathcal{H}_0: \X_1,\ldots,\X_t \sim \mathbf{P}_0 \in \mathcal{D}_0 
\qquad\qquad \mbox{vs.} \qquad\qquad
\mathcal{H}_1: \X_1,\ldots,\X_t \sim \mathbf{P}_1 \in \mathcal{D}_1^{k,\theta}.
$$
Our goal is to distinguish between the possible family of distributions that $\X_1, \X_2, \ldots, \X_t$ is sampled from. The motivation for this problem is that $\mathcal{D}_0$ contains $ \mathcal{N}(0, \I_d)$ and $\mathcal{D}_1$ contains $\mathcal{N}(0, \I_d+\theta\u\u^{\T})$ when $\u$ is a $k$-sparse unit vector. In other words, this hypothesis testing problem is a generalization of the spiked covariance matrix detection problem. Intuitively, $\theta$ corresponds to the signal strength, and $k$ corresponds to the sparsity of the unknown ``high variance'' direction of the alternate hypothesis distribution. Nonetheless, many known algorithms for the spiked covariance problem also hold for this more general problem~\cite{berthet2013complexity}. We note that reductions between planted clique and the spiked covariance version are known~\cite{gao2017sparse}, but we do not know how to implement these reductions efficiently in our query or communication models. Instead, we describe a reduction for the above formulation of the problem.

Let $\mathbb{G}_{m}$ denote the set of graphs on $m$ vertices. With the following reduction, we obtain our main theorem for SCDC. We provide the key details and verify that the reduction holds in our query models, and we refer to~\cite{berthet2013complexity} for the full details about the distributional relationships.

\paragraph{Reduction from \PC~\cite{berthet2013complexity}.}
    For any $\gamma \in (0, 1)$ and a fixed tolerance $\delta \in (0, 1/3)$ (e.g., $\delta = 5\%$), given $(d, t, k) \in R_{\gamma}$, let
    \begin{align*}
        R_{\gamma} = R_0 \cap \{k \geq t^{\gamma}\} \cap \{ t < d \}
    \end{align*}
    and 
    \begin{align*}
        R_0 = \left\{ (d,t,k) \in \mathbb{N}_{+}^3 : 15 \sqrt{\frac{k\log(6ed/\delta)}{t}} \leq 1, k \leq d^{0.49} \right\}
    \end{align*}
     where 0.49 can be any constant $C < 0.5$.
    The randomized reduction $\mathtt{bl}_{d,t,k,m,\kappa}: \mathbb{G}_{2m} \mapsto \mathbb{R}^{d\times t}$ is a procedure defined as follows, where $m,\kappa$ are positive integers with
    $t \leq m < d$  and $k \leq \kappa \leq m$.

For a ($2m$)-vertex graph $G = (V, E)$, which is an instance of \PC problem with a potential clique of size $\kappa$, we first choose $m$ uniformly random vertices $V_{\texttt{left}}$ among $2m$ vertices, and then choose  $t$ uniformly random vertices $V_{\texttt{right}}$ among the remaining $m$ vertices that are not in  $V_{\texttt{left}}$. Make it a bipartite graph by restricting its edges in $E \cap \{ V_{\texttt{left}} \times V_{\texttt{right}}\}$. Then, add $(d-m)$ new vertices to $V_{\texttt{left}}$ and place an edge between every old vertex in $V_{\texttt{right}}$ and each new vertex in $V_{\texttt{left}}$ independently with probability $1/2$. We relabel the left (resp. right) vertices  by a random permutation of $\{1, 2, \ldots, d\}$ (resp. $\{1, 2, \ldots, t\}$). Let $G' = (\{1, 2, \ldots, d\} \times \{1, 2, \ldots, t\}, E')$ denote the resulting bipartite graph, and let $\B$ denote the $d\times t$ adjacency matrix of $G'$. Also, let $\eta_1, \eta_2, \ldots, \eta_t\in\{-1,1\}$ be $t$ i.i.d. Rademacher random variables that are independent of all previous random variables. Define 
    \begin{equation}\eqlab{bl-red}
        \X_i^{(G)} = \eta_i(2\B_i-1) \in \{-1, 1\}^d,
    \end{equation}
    where $\B_i$ is the i-th column of $\B$. By all above steps, we finish the reduction
    \begin{align*}
        \mathtt{bl}_{d,t,k,m,\kappa}(G) = (\X_1^{(G)}, \X_2^{(G)}, \ldots, \X_t^{(G)}) \in \mathbb{R}^{d\times t}.
    \end{align*}

Now that we have described the reduction $\mathtt{bl}$, we explain how to simulate the algorithm. In the query models, we will use the public randomness of the algorithm for this randomized reduction. We next state a simple, yet general, result, which identifies matrix operations that can be simulated in the query models.

\begin{lemma}\lemlab{simulation}
Let $\X$ be a matrix that is a transformation of a matrix $\Y$ after applying one or more of the following operations:
\begin{enumerate}[(i)]
    \item insert a row or column into $\Y$, 
    \item permute the rows or columns of $\Y$, or
    \item for field elements $a,b$, apply $\phi(y) = a y+ b$ to all entries in a row or column of $\Y$, replacing each entry $y$ with $\phi(y)$.
\end{enumerate}
Then, for any query to $\X$ in the $\uMv$, $\Mv$, edge-probe, or linear sketching models, there is a deterministic way to perform a single query to $\Y$ and simulate the original query to $\X$.
\end{lemma}
\renewcommand{\P}{\mathbf{P}}
\begin{proof}
    We first explain the simulation for the general linear sketching model; at the end, we mention the differences for the other models.
    Say that the query algorithm wants to learn 
    $\v^{\T} \vec(\X)$. Our goal is to design a vector $\w$ such that $\v^{\T} \vec(\X)$ can be computed directly from $\w^{\T} \vec(\Y)$.
    
    For (i), assume $\X$ is $\Y$ after inserting a row (the column case is analogous).
    The idea is that we can compute the contribution from the insertion and add this after querying $\Y$. More precisely, 
    let $\Z$ be a matrix with a single non-zero row in the position of the inserted row to $\Y$ with the same entries. Let $\w$ be the vector obtained from $\v$ by deleting the positions in $\v$ corresponding to the inserted row.  As the algorithm knows $\v^{\T}\vec(\Z)$, we have that $\w^{\T}\vec(\Y)$ suffices to compute 
    $\v^{\T}\vec(\X) = \w^{\T}\vec(\Y) +  \v^{\T}\vec(\Z)$.
    
    For (ii), we can permute the entries of $\v$ to obtain $\w$. Precisely, we can determine the permutation matrix~$\P$ such that $\v^{\T}\vec(\X) = \v^{\T}\P\vec(\Y) = (\P^\T\v)^{\T}\vec(\Y)$ and use the query vector $\w = \P^{\T}\v$.
    
    Finally, for (iii), first multiply entries in $\v$ by $a$ for each position corresponding to the modified row or column.
    Let the resulting query vector be $\v_a$. Then, calculate the sum $z$ of entries in $\v_a$ that overlap with the positions in the modified row or column, and add $zb$ to the result of the $\v_a$ query. Overall, by construction we have that $\v_a^{\T} \vec(\Y) + zb = \v^{\T} \vec(\X)$. 
    
    If multiple of these operations are used to transform $\Y$ into $\X$, then we can iteratively apply the above procedures, i.e., we can simulate any query with a single other query.

    We have described the simulation for the general linear sketching model. The same strategy works for the $\F_2$ sketching model. For the edge-probe model, note that a single entry in $\X$ depends on only a single entry of $\Y$ even after any of the three allowed operations. In the $\uMv$ model, the difference for (iii) is that we can rescale rows via the left query vector (or columns via the right query vector) by multiplying the relevant entry by $a$. For the $\Mv$ model, the main difference is that when we modify a row of $\Y$, we have to compute the contribution to $\X \v$ and add this to obtain the correct query, similar to (iii) above. Columns in (iii) can be handled by updating the query vector.  
\end{proof}

\begin{proposition} \proplab{br-error}
There exists a constant $\delta \in (0,1)$ such that the following holds.
If there is a query algorithm $\mathcal{A}$ for the \SCDC problem that makes $q$ queries and has success probability $1-\delta$, then there is a query algorithm $\mathcal{A}'$ for the \PC problem that makes $q$ queries and has success probability $1-\Theta(\delta)$. This holds for algorithms in the $\uMv$, $\Mv$, edge-probe, and linear sketching models.
\end{proposition}
\begin{proof}
Let $\mathcal{A}$ be an algorithm for \SCDC in a query model that we consider.
Given an instance of planted clique, we have a graph $G$. We can apply the randomized reduction to produce a $d \times t$ matrix $\mathtt{bl}(G)$ composed of columns $(\X_1^{(G)}, \X_2^{(G)}, \ldots, \X_t^{(G)})$.  
We claim that we can use \lemref{simulation} to take the algorithm $\mathcal{A}$ for \SCDC and derive an algorithm $\mathcal{A}'$ for the \PC problem on $G$ via the procedure $\mathtt{bl}(G)$. Here, $\Y$ corresponds to the adjacency matrix of $G$ and $\X$ corresponds to $\mathtt{bl}(G)$, the \SCDC input.
First, even though $\mathtt{bl}$ is a randomized procedure, we have used the randomness of the query algorithm, and hence we know the transformation. The linear transformation in \Eqref{bl-red} is covered by  part (iii) of \lemref{simulation}. Then, in the procedure $\mathtt{bl}(G)$, we are sampling vertices, permuting vertex labels, and adding edges connected to new vertices in the graph, which are covered by the three parts of the lemma.


It remains to show  that if $\mathcal{A}$ has constant success probability for \SCDC, then $\mathcal{A}'$ has constant success probability for \PC. We sketch this argument, 
since it follows from Theorem 7 and Lemma 8 in~\cite{berthet2013complexity}. They work in a more general model, where they consider families of statistical tests $\psi = \{ \psi_{d,t,k}\}$ for \SCDC and tests $\xi = \{\xi_{m, \kappa}\}$ for \PC on $2m$-vertex graphs. In our models, this corresponds to the query algorithms by considering a test to be a query algorithm that queries the matrix and then outputs a binary variable, where $0$ corresponds to the null hypothesis, and $1$ corresponds to the alternate hypothesis. 

To state their result, fix $\alpha \in [1, 2), \gamma \in (0, \frac{1}{4-\alpha})$, and define $a = 2\gamma, b = 1 - (2-\alpha)\gamma$.  Their result says that for any $\tau > 0$, there exists a constant $L > 0$, such that the following holds. For $(d,t,k) \in R_{\gamma}$, there exist $\kappa, m$ such that $(2m)^{a/2} \leq \tau \kappa \leq (2m)^{b/2}$, a random transformation $\mathtt{bl} = \{\mathtt{bl}_{d,t,k,m,\kappa}\},$ $\mathtt{bl}_{d,t,k,m,\kappa}: \mathbb{G}_{2m} \mapsto \mathbb{R}^{d \times t}$, and distributions $\mathbf{P}_0 \in \mathcal{D}_0$, $\mathbf{P}_1 \in \mathcal{D}_1^{k, L\theta_{\alpha}}$ such that the following holds.  For shorthand, we use the notation $\mathbf{P}_0(f=1)$ for a test $f$ to mean the probability that the output is 1 when an instance is drawn from $\mathbf{P}_0$, i.e., the error probability of NO instances (and analogously for $\mathbf{P}_1(f=0)$). We also use $\vee$ to mean $\max$. Then, their Theorem 7 says that there exists a constant $\delta$ such that for any family of tests $\psi = \{ \psi_{d,t,k}\}$, we have 
    \begin{align*}
        \mathbf{P}_0^{\otimes t}(\psi_{d,t,k} = 1) \vee \mathbf{P}_1^{\otimes t}(\psi_{d,t,k} = 0) \geq \mathbf{P}_0^{(G)}(\xi_{m, \kappa}(G) = 1) \vee \mathbf{P}_1^{(G)}(\xi_{m, \kappa}(G) = 0)  - \frac{\delta}{5},
    \end{align*}
    where 
    \begin{align*}
        \xi_{m, \kappa} = \psi_{d,t,k} \circ \mathtt{bl}_{d,t,k,m,\kappa} \text{ and } \theta_{\alpha} = \sqrt{\frac{k^{\alpha}}{t}}.
    \end{align*}
\end{proof}    
    
Using \propref{br-error}, we see that if we have 
    any query algorithm for \SCDC with error probability at most $\delta'$, then we can derive a query algorithm for \PC with error probability at most $\delta' + \delta/5$. 
    In other words, 
    a constant success probability query algorithm for  \SCDC implies one for  \PC with the same query complexity.
    This allows us to derive the query complexity lower bounds in the following theorem.

\renewcommand{\hat}{\widehat}

\begin{theorem}
    Given any $\alpha \in [1, 2)$, $\gamma \in (0, \frac{1}{4-\alpha})$ and any $(d,t,k) \in R_{\gamma}$, for the \SCDC problem with input matrix $\X = [\X_1, \X_2, \cdots, \X_t]$ when $d = \Theta(t)$, $k = \Theta(t^\gamma)$ and $\theta = \Theta\left(\sqrt{\frac{k^\alpha}{t}}\right)$, any algorithm that succeeds with constant success probability requires 
    \begin{itemize}
        \item $\wt \Omega\left(\frac{k^4}{t^2\theta^4}\right)$ queries in the general linear sketching or $\uMv$ models,
        \item $\wt \Omega\left(\frac{k^4}{t^3\theta^4}\right)$ queries in the $\Mv$ model,
        \item $\Omega(k^2/\theta^2)$ queries in the edge-probe or $\F_2$ sketching models.
    \end{itemize}
\end{theorem}
\begin{proof}
    By using the reduction above and applying \propref{br-error}, we can use an algorithm for \SCDC to solve the \PC problem on $m$-vertex graphs with $\kappa$ being the potential planted clique size. Translating between the parameters via the reduction, plugging in $k = \Theta\left(t^\gamma\right)$, $\kappa = O\left(t^{\frac{1}{2} - \frac{(2-\alpha)\gamma}{2}}\right)$, $\theta = \Theta\left(\sqrt{\frac{k^\alpha}{t}}\right)$ and $t = \Theta(m) = \Theta(d)$, yields the following query lower bounds:
    \begin{itemize}
    \item 
    Using the lower bound of $\wt \Omega(m^2/\kappa^4)$ queries for the \PC problem in the general linear sketching and $\uMv$ models when $\kappa = o(\sqrt{m})$ from~\corref{coro_PC_linear_sketching}, we get a query lower bound for \SCDC of 
    $$\wt \Omega\left(\frac{m^2}{\kappa^4}\right) 
    = 
    \wt \Omega\left(\frac{m^2}{t^{2-2(2-\alpha)\gamma}}\right) 
    = 
    \wt \Omega\left(t^{2(2-\alpha)\gamma}\right) 
    = 
    \wt \Omega\left(\frac{t^{4\gamma}}{t^{2\alpha \gamma}}\right) 
    =  
    \wt \Omega\left(\frac{k^{4}}{k^{2\alpha}}\right) 
    =  
    \wt \Omega\left(\frac{k^4}{t^2\theta^4}\right),$$
    since $m^2 = \Theta(t^2)$ and $t^{\gamma} = \Theta\left(k\right)$, and in the final equality, we use that $\frac{t^2}{k^{2\alpha}} = \Theta\left(\frac{1}{\theta^{4}}\right)$.
    \item Using the lower bound of $\wt \Omega(m/\kappa^4)$ queries for the \PC problem in the $\Mv$ model when $\kappa = o(\sqrt{m})$ (direct corollary of \corref{coro_PC_linear_sketching}), we get a query lower bound for \SCDC of 
    $$
    \wt \Omega\left(\frac{m}{\kappa^4}\right)
    =
    \wt \Omega\left(\frac{mt}{t\kappa^4}\right)
     =
    \wt \Omega\left(\frac{m^2}{t\kappa^4}\right)
    =
    \wt \Omega\left(\frac{k^4}{t^3\theta^4}\right),
    $$
    using the above calculations and the fact that $m = \Theta(t)$.
    \item Using the lower bound of $\Omega(m^2/\kappa^2)$ for the \PC problem in the edge-probe and $\F_2$ sketching models when $\kappa = o(\sqrt{m})$ from \corref{xor-pc}, we get a query lower bound for \SCDC  of
    $$
    \Omega\left(\frac{m^2}{\kappa^2}\right)
    =
    \Omega\left(\frac{t^2}{\kappa^2}\right)
    =
    \Omega\left(\frac{t^2}{t^{1-(2-\alpha)\gamma}}\right)
     =
    \Omega\left(\frac{t \cdot t^{2\gamma}}{t^{\alpha\gamma}}\right)
    =    
    \Omega\left(\frac{t \cdot k^2}{k^{\alpha}}\right)
    =
    \Omega\left(\frac{k^2}{\theta^2}\right),
    $$
    where the final equality uses that $\frac{t}{k^{\alpha}} = \Theta\left(\frac{1}{\theta^{2}}\right)$.
    \end{itemize}
        

\end{proof}

%% file: conclusion.tex
Motivated by  understanding statistical-computational trade-offs, we addressed a variety of related average-case communication complexity problems. To this end, we developed a generic reduction technique that preserves the distribution of graph problems that can be defined in terms of planted subgraphs. Specifically, we proved new lower bounds for the planted clique problem and three variants: the bipartite version, the semi-random version, and the promise version. For the $\F_2$ sketching model (and edge-probe model as a special case), we obtained tight bounds on the query complexity. For the more general linear sketching model, we also proved new lower bounds for these problems, and we demonstrated a lower bound for the hidden hubs problem. Finally, we provided lower bounds for a variant of the \SPCA problem. 

Looking forward, our techniques may be useful for developing a more general theory of average-case communication complexity. Indeed, the next step could be to explore the natural analogues of other statistical problems that have been reduced to planted clique~\cite{berthet2013complexity, brennan2019average, brennan2020reducibility, kunisky2019notes}. A more concrete direction is to close the gaps in \tabref{results}. For example, in the linear sketching model we establish that the query complexity of the \PC problem is between $\widetilde \Omega(n^2/k^4)$ and $\widetilde O(n^2/k^2)$. Similarly, for the planted $r \times s$ biclique problem (\BPC), the complexity is between $\widetilde \Omega(n^2/(rs)^2)$ and $\widetilde O(n^2/(r^2s))$ when $r \gg \sqrt{n} \log n$. Another direction could be to determine a non-linear query model where a bound of $\Theta(n/k^2)$ can be derived for the $\PC$ problem when $k = o(\sqrt{n})$ and when the query only reveals $O(\log n)$ bits (compared to the $\Mv$ model, which reveals $O(n \log n)$ bits).